\documentclass[11pt,a4paper]{article}

\usepackage[utf8]{inputenc}
\usepackage[T1]{fontenc}

\usepackage{amsmath,amssymb,amsthm}
\usepackage{graphicx}
\usepackage{algorithm}
\usepackage{algpseudocode}
\usepackage[numbers]{natbib} 
\usepackage{hyperref}
\usepackage{tikz}
\usepackage{subcaption}
\usetikzlibrary{calc}

\usepackage[a4paper,margin=1in]{geometry}

\usepackage{xspace}
\newcommand{\BID}{\textsc{BID}\xspace}

\theoremstyle{plain}
\newtheorem{theorem}{Theorem}[section]
\newtheorem{corollary}{Corollary}
\newtheorem{lemma}[theorem]{Lemma}
\newtheorem{claim}[theorem]{Claim}
\newtheorem{proposition}[theorem]{Proposition}

\theoremstyle{definition}
\newtheorem{definition}[theorem]{Definition}

\title{On the Instance Optimality of Bidirectional Dijkstra's Algorithm}

\author{
Matic Požar\\
UP FAMNIT, University of Primorska\\
Koper, Slovenia\\
\texttt{matic.pozar@upr.si}
}

\date{}

\begin{document}

\maketitle

\begin{abstract}
We study instance optimality for the shortest $st$-path problem
on positively weighted graphs, focusing on how graph structure
and query access affect the information needed to certify a
shortest-path distance. For undirected simple graphs in the
standard adjacency-list model, we establish instance-sensitive
lower bounds against every universally correct randomized
algorithm. A quantitative reservoir bound relates bidirectional
Dijkstra's instance-optimality ratio to the size of its
exploration and the vertex count. A complementary tensor-graph
construction turns pairs of low-query explored edges into
vertex-cover constraints; explored cores, degree moments,
exploration orientations, and auxiliary degree profiles yield
further sufficient conditions. We also obtain guarantees from
exceptional high-degree vertices and individually
weight-critical edges, and show that a superconstant
improvement over bidirectional Dijkstra requires almost all
explored edges on both sides to concentrate around sublinear
vertex sets, after discarding a vanishing fraction of the edges.

We additionally investigate a \emph{local-discovery} query
model in which algorithms initially know only the terminal
handles, cannot access undiscovered vertices or the total
vertex count, and may still make degree queries at discovered
vertices. In this model, corrected bidirectional Dijkstra
is constant-factor instance-optimal on \emph{all} directed
and undirected positive-weight simple graphs, even with
incoming and outgoing degree queries. The theorem
holds for every fixed adjacency ordering and hence
also under order-oblivious averaging. Finally, we revisit the
correction to the original bidirectional Dijkstra update
rule and retain a self-contained proof of the corrected
multigraph guarantee. Together, these results isolate
both structural and information-access conditions
underlying instance optimality.
\end{abstract}

\section{Introduction}
\label{sec:introduction}

The shortest $st$-path problem asks for a minimum-length path
between designated vertices $s$ and $t$ of a weighted graph.
Dijkstra's algorithm~\cite{dijkstra1959} solves this problem
for nonnegative edge weights, with standard implementations
running in $O(|E|+|V|\log |V|)$
time~\cite{dijkstra_Fibonacci}. Instance optimality asks
a different question: whether an algorithm uses, on every
individual input, at most a constant factor more queries
than any competing algorithm that is correct on all
admissible inputs~\cite{intance_optimality,original_instance_optimality}.
The answer depends not only on the graph, but also on what
the query oracle reveals. In particular, obtaining a candidate
shortest path and certifying that no shorter path exists
can require very different amounts of information.

Haeupler, Hlad\'{\i}k, Rozho\v{n}, Tarjan, and
T\v{e}tek~\cite{BID_DIJKSTRA_REV2026} established instance
optimality of bidirectional Dijkstra on positively weighted
directed and undirected multigraphs. For simple graphs,
however, the guarantee depends on the graph and query
models. We study which structural properties of individual
bidirectional searches force every correct competing
algorithm to make comparably many queries. We also
investigate a local-discovery oracle in which the absence
of global vertex access restores instance optimality
on directed and undirected simple graphs, regardless
of adjacency-list ordering.

The precise access model is essential even for unidirectional
shortest-path search. Haeupler et al.~\cite[Theorem~2]{BID_DIJKSTRA_REV2026}
prove an instance-optimality result for unidirectional Dijkstra
against deterministic algorithms in a restricted directed model
in which queries reveal only successive out-neighbors of
already discovered vertices. By contrast, in an undirected
adjacency-list model, querying the neighborhood of $t$
also reveals edges entering $t$. We give a family of undirected
instances on which this additional information allows a
correct algorithm to certify an $st$-path using asymptotically
fewer queries than unidirectional Dijkstra.

Whether the multigraph guarantee extends to simple graphs
is a separate question. In a multigraph, an indistinguishability
argument may replace explored edges with a shorter connection
even if the replacement creates parallel edges. In a simple
graph, an intended replacement edge may already exist,
obstructing that modification. Bertram, Jensen, Thorup,
Wang, and Yan~\cite{BertramJensenThorupWangYan2026},
in parallel and independent work, investigate this issue
across directed and undirected graphs, order-dependent
and order-oblivious adjacency lists, degree-query access,
and different weight domains. In particular, they establish
instance optimality for undirected simple graphs in the
order-oblivious setting, while showing that it can fail
in the order-dependent setting; their work also gives
worst-case bounds and separating constructions for
directed graphs. Thus, no unrestricted instance-optimality
theorem holds across all simple-graph query models.

\paragraph{Order-dependent and order-oblivious comparisons.}
For a weighted graph $G$ with terminals $s,t$, let $\pi$
specify its vertex identifiers and adjacency-list orderings.
In the \emph{order-dependent} model, the representation
$(G,\pi)$ is fixed: a universally correct competitor may be
tailored to favorable indices, but must verify its answer.
In the \emph{order-oblivious} model, the representation is
randomized and the competitor is chosen before the
randomization. In particular, its benchmark is
$\inf_{A\ {\rm good}}\mathbb E_\pi[T_A(G,\pi)]$, not
$\mathbb E_\pi[\inf_{A\ {\rm good}}T_A(G,\pi)]$.

The latter model is graph-invariant and already admits
unconditional instance optimality for undirected
positive-weight simple graphs by
Bertram et al.~\cite{BertramJensenThorupWangYan2026}.
Our certificates address the order-dependent question,
which also permits unfavorable fixed representations.
A ratio bound valid uniformly over representations
remains valid after averaging; a condition verified
only for one ordering need not.

\paragraph{Motivation and scope.}
Worst-case separations do not reveal which graph
structures make particular shortest-path queries hard.
We therefore seek certificates based on BID's actual
exploration, rather than only global graph parameters.
Testing how frequently these certificates apply to
real routing and network data remains an empirical
question. Except for the multigraph and local-discovery
theorems, the structural results below concern
undirected simple graphs with strictly positive weights,
in a fixed, known vertex universe. Notation is given
in Section~\ref{sec:Preliminaries}.

\paragraph{Main contributions.}
\begin{enumerate}
    \item \textbf{Instance optimality under local discovery.}
    In a model where algorithms initially know only the terminal
    handles, discover other vertices exclusively through adjacency
    queries, and are not given the vertex count, we show that
    corrected BID satisfies
    \[
        T_{\mathrm{BID}}(G,s,t)\le30\,T_A(G,s,t)
    \]
    on directed and undirected positive-weight simple graphs,
    even with incoming and outgoing degree and adjacency queries.
    The bound holds for every fixed adjacency-list ordering and
    hence also in the order-oblivious model. It relies on local
    discovery and does not imply unconditional instance optimality
    in the fixed-vertex-universe model.

    \item \textbf{Quantitative reservoir bound.}
    For undirected simple graphs in the standard query model, a
    degree-preserving indistinguishability construction yields a
    lower bound for every randomized algorithm succeeding with probability
    at least $0.9$ on every admissible input and, consequently,
    \[
        \rho_{\mathrm{BID}}(G,s,t)
        \le 30+\frac{5(S+1)}{n}.
    \]
    In particular, $S=O(n)$ guarantees a constant ratio, including
    on graph classes with $|E|=O(n)$. The bound quantifies an
    instance-sensitive sparse-exploration regime.

    \item \textbf{Tensor reduction and optimizable explored cores.}
    We associate pairs of low-query first-exploration edges with an auxiliary
    tensor graph whose vertex cover is constrained by the cross edges
    present in the input. Its lower bounds apply to every randomized
    algorithm succeeding with probability at least $0.9$ on every
    admissible input, without assuming any particular query strategy. For any nonempty first-exploration subsets $F_s\subseteq E_s^{(1)}$
    and $F_t\subseteq E_t^{(1)}$, writing
    $m_i=|F_i|$ and $D_i=\Delta(F_i)$, we show that
    \[
        Q\ge\frac{1}{16}\min\!\left\{
            m_s,m_t,\frac{m_sm_t}{D_sD_t}
        \right\}.
    \]
    Maximizing the expression inside the minimum over explored
    subsets gives a polynomial-time computable certificate through
    cardinality $b$-matching. We exhibit a family with exceptional
    hubs on which this core criterion establishes instance optimality
    while the untrimmed degree and reservoir criteria do not.

    \item \textbf{Moment, orientation, and degree-profile refinements.}
    Let $P_s,P_t$ be the sums of squared explored degrees on the two
    sides. We establish
    \[
        \rho_{\mathrm{BID}}
        =O\!\left(\max\left\{1,\frac{P_sP_t}{S^3}\right\}\right),
    \]
    and obtain the analogous, orientation-sensitive bound with
    $P_sP_t$ replaced by $\Psi$, the squared-degree sum of the
    oriented auxiliary tensor graph. A further degree-profile
    certificate gives $Q\ge\min\{S/128,J/4\}$, where $J$ is the
    number of highest-degree auxiliary vertices needed to accumulate
    total degree at least $S^2/256$. These bounds retain more
    information than global maximum-degree estimates.

    \item \textbf{Additional structural and metric regimes.}
    For undirected simple graphs, if at most $h_D$ vertices
    have degree greater than $D$, we obtain
    $\rho_{\mathrm{BID}}=O(1+h_D+D^2/n)$.
    We also show that every universally correct competing
    algorithm must query each weight-critical edge with
    probability at least $0.8$, giving the independent
    lower bound $Q\ge(4/5)|\mathcal C_{s,t}|$.
    Both conditions can establish instance optimality
    beyond a description by global maximum degree alone.

    \item \textbf{Necessary concentration for superconstant gaps.}
    If a correct algorithm has $Q=o(S)$, then after removing $o(S)$
    explored incidence records from each side, the remaining
    forward and backward graphs each admit a vertex cover of size
    $o(n)$. Quantitatively, when $Q<S/128$, the covers may be chosen
    with sizes at most $256(n-1)Q/S$, after deleting at most $8Q$
    records per side. Thus, a superconstant improvement over BID
    requires strong concentration of explored edges on both sides,
    though this condition is not sufficient for such an improvement.

    \item \textbf{Corrected bidirectional Dijkstra on weighted multigraphs.}
    We describe the minimal relaxation of the path-update rule in the
    initial implementation of Haeupler et al.~\cite{BID_DIJKSTRA_REV2026} and
    give an independent, self-contained proof of constant-factor
    instance optimality for positive-weight multigraphs, both directed
    and undirected. We also exhibit a family separating the initial
    implementation from the corrected one.
\end{enumerate}

\paragraph{Correction to previous versions and publication history.}
Haeupler et al.'s September 2026 revision
(version~4)~\cite{BID_DIJKSTRA_REV2026} removed the
extra backward-closure requirement in their original
BID update rule, independently identified here and
by Bertram et al.~\cite{BertramJensenThorupWangYan2026}.

We also formally retract an assertion in earlier versions
of this manuscript that Haeupler et al.'s restricted
unidirectional theorem~\cite[Theorem~2]{BID_DIJKSTRA_REV2026}
was incorrect. Our purported counterexample implicitly
used undirected adjacency access (or incoming edges at $t$),
which their outgoing-only directed oracle forbids.
The undirected separation is valid, but does not
contradict that theorem.

\paragraph{Relation to concurrent work.}
Bertram et al.~\cite{BertramJensenThorupWangYan2026} independently
studied instance optimality on simple graphs in a preprint posted
on August 25, 2026, two days before the initial preprint of
the present work (August 27, 2026). Their Theorems~8.8 and~8.9
explicitly establish the global bounds
$O(\min\{\lceil |E(G)|/n\rceil,n^{1/3}\})$ and
$O(\lceil\Delta^{2/3}/n^{1/3}\rceil)$ on the
instance-optimality ratio in the order-dependent undirected model.
In particular, these imply constant-factor instance optimality
when $|E(G)|=O(n)$ or $\Delta=O(\sqrt n)$.
The reservoir and explored-degree arguments in the initial
version of our preprint independently yielded corresponding
constant-factor regimes, including the execution-sensitive
conditions $S=O(n)$ and $\Delta^2=O(S)$, together with
the reservoir case. The original reservoir argument also
implies $\rho_{\mathrm{BID}}=O(1+S/n)$, although this
quantitative ratio was not stated explicitly there.
Bertram et al.\ were the first of these concurrent preprints
to state explicit parameterized \emph{global} ratio guarantees.
Their Lemma~8.11 also implies an $O(1+S/n)$ bound, although
that execution-sensitive form is not separately stated as a theorem. Their
$O(1+(\Delta^2/n)^{1/3})$ bound is quantitatively sharper
than the global guarantee obtained by combining our
elementary reservoir and maximum-degree tensor estimates.
The contribution developed here is a family of universal,
execution-sensitive lower bounds against \emph{every}
randomized algorithm that succeeds with probability at least
$0.9$ on every admissible input, including optimizable cores,
degree moments, and orientation-sensitive tensor certificates.
The concurrent development should not be interpreted as
either preprint deriving its results from the other.
These results do not establish how frequently the refined
conditions hold on real-world graphs, an empirical question
left for subsequent evaluation.

The remainder of the paper specifies the query models, revisits the
restricted unidirectional setting and corrected multigraph algorithm,
and then develops the local-discovery result for directed and
undirected simple graphs and the standard-model structural
results for undirected simple graphs.

\section{Preliminaries} \label{sec:Preliminaries}

We write $G=(V,E)$ for a weighted graph and $n=|V(G)|$
for its number of vertices. Unless otherwise stated, weights are
strictly positive. General multigraphs may contain loops and
parallel edges; simple graphs have neither. We use the
adjacency-list query model of~\cite{BID_DIJKSTRA_REV2026},
as discussed in Chapter~10 of~\cite{query_model}.
In the standard model the identifiers $1,\ldots,n$ are known
and may be queried directly. Degree and indexed adjacency-list
queries have unit cost; an indexed query returns the adjacent
vertex and edge weight. In directed graphs, separate incoming
and outgoing lists and degree queries are available.
Our local-discovery model is defined in
Section~\ref{subsec:local-discovery}.

\paragraph{Global conventions.}
Throughout, $E_s$ and $E_t$ denote the incidence records
queried by the forward and backward searches of corrected
\BID{}, respectively, on the input being studied. These
records need not correspond to distinct underlying edges:
an undirected edge $\{u,v\}$ occurs in the adjacency
list of each endpoint and may be scanned twice by one
search. We write
\[
S=|E_s|+|E_t|.
\]
The searches alternate edge queries, so
$\bigl||E_s|-|E_t|\bigr|\le1$.
The total $S$ counts \emph{all} queried incidence entries,
not merely distinct physical edges; separate degree queries
are included in the full query complexity
$T_{\BID}(G,s,t)$. For any universally correct competing
algorithm $A$, let $T_A(G,s,t)$ denote its expected number
of queries; when $A$ is fixed, abbreviate this to $Q$.
We use
\[
\tau(G,s,t)=\inf_{A\ \mathrm{correct}}T_A(G,s,t),
\qquad
\rho_{\mathrm{BID}}(G,s,t)=
\frac{T_{\BID}(G,s,t)}{\tau(G,s,t)}.
\]
The infimum is taken over algorithms correct on
\emph{every} admissible input in the query model under
consideration, with success probability at least $0.9$.
Unless qualified, lower bounds on $Q$ apply to every
such competitor, including randomized algorithms.

\paragraph{First explorations and query probabilities.}
For $i\in\{s,t\}$ let $E_i^{(1)}\subseteq E_i$
retain only the first queried incidence record of
each underlying edge in that search. Then
$|E_i^{(1)}|\ge |E_i|/2$. In the undirected
simple-graph tensor construction, a first-exploration
record $u\to v$ is oriented from its issuing vertex
to its neighbor, giving $d(s,u)\le d(s,v)$
forward and $d(u,t)\le d(v,t)$ backward.
The runtime still counts every record in $E_s,E_t$.
For an underlying edge or arc $e$, define $p_A(e)$
as the probability that $A$ reveals it through
\emph{any} of its incidence entries. In particular,
for an undirected edge, querying either endpoint
counts as revealing $e$. For any family $F$ of
records from one BID exploration,
\[
\sum_{e\in F}p_A(e)\le2Q,
\]
because each underlying edge occurs at most twice in
$F$, whereas each query by $A$ reveals at most one.
If $F$ consists only of first explorations,
the factor $2$ is unnecessary. When different
pairs of retained records generate the same auxiliary
tensor edge, the tensor construction retains multiplicity.

Now we will briefly describe the \textbf{shortest $st$-path problem}. We are given a weighted multi-graph $G$. The edge weights are determined by a weight function $\ell:E\rightarrow\mathbb{R}_{>0}$. In our case it will be important to consider strictly positive edge weights. This edge weight function $\ell$ gives rise to a distance function $d : V\times V \rightarrow \mathbb{R}_{\ge 0}\cup\{+\infty\}$ that maps any pair of nodes $(u,v)$ into the shortest distance from $u$ to $v$ (possibly $+\infty$ if no path exists). The shortest $st$-path problem asks for a shortest path $P$ from $s$ to $t$ together with its length $d(s,t)$, or a distinguished $\bot$ (no-path) output when no such path exists. Throughout the algorithmic statements, the case $s=t$ is handled immediately by returning the empty path together with length $0$.

\textbf{Dijkstra's algorithm} \cite{dijkstra1959} accepts a weighted multigraph with nonnegative edge weights, and a starting node $s$. It then computes the shortest distances by storing for each node $u$ the shortest path computed thus far $\hat d(s,u)$. In the beginning only the starting node $s$ is \emph{open}. The main loop of the algorithm stops when there are no more open nodes. In each iteration of the main loop the open node $u$ with the smallest computed distance $\hat d(s, u)$ is chosen and set \emph{closed}. Then for each of the open neighbors $v$ of $u$ their computed distance is updated as $\hat d(s,v)= \min (\hat d(s,v), \hat d(s,u)+\ell(u,v))$. Once a node has been closed its computed distance equals the true distance.

\textbf{Instance optimality}~\cite{original_instance_optimality} is a strong measure of algorithmic complexity that essentially gives guarantees regarding any possible instance of the algorithm.
\begin{definition}[Algorithm Correctness]
An algorithm $A$ is said to be correct if on any input it returns the correct answer with probability at least 0.9.
\end{definition}
Naturally, the 0.9 cutoff is arbitrary and can change in different contexts.
\begin{definition}[Instance Optimality] \label{def:instanc_opt}
An algorithm $A$ is instance optimal if it is correct and there exists a constant $c=O(1)$ such that for any input $x$ and for any correct algorithm $A'$ it holds that the expected complexity $T_A(x)$ and $T_{A'}(x)$ of $A$ and $A'$ on $x$ satisfy
$$T_A(x)\le c\cdot T_{A'}(x) \text{.}$$
\end{definition}

\subsection{The Role of the Query Model in Unidirectional Instance Optimality}
\label{sec:Dijkstra}

The importance of the query model is already apparent in
unidirectional shortest-path search. Haeupler et al.~
\cite[Theorem 2]{BID_DIJKSTRA_REV2026} establish instance optimality
of Dijkstra's algorithm against deterministic competitors
in a restricted directed model, where queries reveal only
successive out-neighbors of previously discovered vertices.
In particular, the in-neighborhood of $t$ cannot be
enumerated. Two observations sharpen the role of these
restrictions. Allowing \emph{outdegree} queries at discovered
vertices does not defeat the edge-redirection argument of
Theorem~2 in the directed multigraph setting: replacing an
unqueried arc $uv$ by $ut$ preserves every outdegree.
Together with constant-cost checks for empty outgoing
lists, this suggests that constant-factor instance
optimality survives, although the exact query count
can change.

By contrast, \emph{indegree} queries can invalidate
the guarantee. Take $k$ vertices $U$, all weight-$1$
arcs $s\to u$ for $u\in U$, all weight-$1$ arcs
between distinct vertices of $U$ in both directions,
and a weight-$3$ arc $s\to t$ appearing first in
$s$'s outgoing list. There are no other arcs entering
$t$. Dijkstra explores $\Theta(k^2)$ outgoing
entries before closing $t$, whereas a universally
correct algorithm can query $\operatorname{InDegree}(t)$,
then the first outgoing arc of $s$, and return
$s\to t$ when these certify that it is the only
arc entering $t$, falling back to Dijkstra otherwise.
The shortcut costs only two queries on the example.
Thus even incoming \emph{degree information}, without
an incoming-edge enumeration oracle, can invalidate
unidirectional instance optimality.

We illustrate the complementary effect of undirected
adjacency access by considering the corresponding
problem on undirected graphs.

Consider an algorithm $\mathrm{CHEAP}$ that first examines
the adjacency lists of $s$ and $t$, recording the minimum
incident edge weights $b_s$ and $b_t$, respectively.
If $s$ and $t$ are not adjacent, the algorithm checks
whether they have a common neighbor $m$ satisfying
\[
    \ell(s,m)+\ell(m,t)=b_s+b_t.
\]
If such a vertex exists, it returns the path $(s,m,t)$;
otherwise, it invokes any correct shortest-path algorithm.
This procedure is correct because every $st$-path containing
at least two edges has length at least $b_s+b_t$.
The preliminary check requires only
$O(\deg(s)+\deg(t)+1)$ queries.

To illustrate the difference, consider the graph family
$G_D$ shown in Figure~\ref{fig:G3}. It consists of two
disjoint cliques $K_{D+1}$, containing $s$ and $t$,
respectively, connected by an additional vertex $m$
through the edges $(s,m)$ and $(m,t)$. All edges
incident to $s$ or $t$ have weight $1$, while the
remaining edges have arbitrary positive weights.
The graph has $n=2D+3$ vertices, and its unique
shortest $st$-path is $(s,m,t)$, of length $2$.

\begin{figure}[t]
    \centering
    \begin{tikzpicture}[
    scale=1,
    transform shape,
    dot/.style={circle, fill=black, inner sep=1.7pt},
    clique/.style={draw=gray!65, fill=gray!7, line width=0.7pt},
    internal/.style={draw=gray!65, line width=0.65pt},
    lab/.style={font=\small}
]


\coordinate (s)  at (-1.35, 0.00);
\coordinate (s1) at (-2.62, 0.34);
\coordinate (s2) at (-2.62,-0.34);
\coordinate (sD) at (-2.12, 0.54);

\coordinate (t)  at ( 1.35, 0.00);
\coordinate (t1) at ( 2.62, 0.34);
\coordinate (t2) at ( 2.62,-0.34);
\coordinate (tD) at ( 2.12, 0.54);

\coordinate (m) at (0,0);


\draw[clique] (-2.15,0) ellipse [x radius=1.3, y radius=1.1];
\draw[clique] ( 2.15,0) ellipse [x radius=1.3, y radius=1.1];


\begin{scope}
    \clip (-2.15,0) ellipse [x radius=1.03, y radius=0.86];

    \draw[internal] (s)--(s1);
    \draw[internal] (s)--(s2);
    \draw[internal] (s)--(sD);

    \draw[internal] (s1)--(s2);
    \draw[internal] (s1)--(sD);
    \draw[internal] (s2)--(sD);
\end{scope}


\begin{scope}
    \clip (2.15,0) ellipse [x radius=1.03, y radius=0.86];

    \draw[internal] (t)--(t1);
    \draw[internal] (t)--(t2);
    \draw[internal] (t)--(tD);

    \draw[internal] (t1)--(t2);
    \draw[internal] (t1)--(tD);
    \draw[internal] (t2)--(tD);
\end{scope}



\draw[black, line width=1.15pt] (s)--(m)--(t);


\node[dot] at (s) {};
\node[dot] at (s1) {};
\node[dot] at (s2) {};
\node[dot] at (sD) {};

\node[lab, below=3pt] at (s) {$s$};
\node[lab, left=3pt] at (s1) {$s_1$};
\node[lab, left=3pt] at (s2) {$s_2$};
\node[lab, above=3pt] at (sD) {$s_D$};

\node[lab, fill=gray!7, inner sep=0.8pt]
    at (-2.24,-0.11) {$\cdots$};

\node[lab] at (-2.15,-1.35) {$K_{D+1}$};


\node[dot] at (t) {};
\node[dot] at (t1) {};
\node[dot] at (t2) {};
\node[dot] at (tD) {};

\node[lab, below=3pt] at (t) {$t$};
\node[lab, right=3pt] at (t1) {$t_1$};
\node[lab, right=3pt] at (t2) {$t_2$};
\node[lab, above=3pt] at (tD) {$t_D$};

\node[lab, fill=gray!7, inner sep=0.8pt]
    at (2.24,-0.11) {$\cdots$};

\node[lab] at (2.15,-1.35) {$K_{D+1}$};


\node[dot] at (m) {};
\node[lab, above=5pt] at (m) {$m$};

\end{tikzpicture}
    \caption{The graph family $G_D$. Two disjoint cliques
    $K_{D+1}$, containing $s$ and $t$, respectively, are
    connected by the path $s-m-t$. Only representative
    clique vertices and their edges are shown.}
    \label{fig:G3}
\end{figure}

On $G_D$, $\mathrm{CHEAP}$ identifies the shortest path
after examining only the adjacency lists of $s$ and $t$,
using $O(D)=O(n)$ queries. In contrast, unidirectional
Dijkstra must close every vertex $s_i$ in the clique
containing $s$, since
\[
    d(s,s_i)=1<d(s,t)=2.
\]
Each such vertex has degree $D$, so Dijkstra performs
$\Omega(D^2)=\Omega(n^2)$ queries. Thus, unidirectional
Dijkstra is not instance optimal in the undirected
adjacency-list model, despite being instance optimal
in the restricted directed model of
Haeupler et al.~\cite{BID_DIJKSTRA_REV2026}.

The distinction lies in the information available at
the target vertex: examining the neighborhood of $t$
provides a lower bound on the weight of the final edge
of any $st$-path in an undirected graph, whereas querying
out-neighbors of $t$ provides no such information in a
general directed graph. Notably, $\mathrm{CHEAP}$ only queries the initially
known vertices $s,t$; it does not need direct access
to undiscovered vertices. If sequential adjacency
queries signal when a list is exhausted, it also
does not need separate degree queries. Thus even
purely local, degree-free access does not restore
instance optimality of unidirectional Dijkstra on
undirected graphs. This distinction also matters
for Definition~\ref{def:local-discovery} below.

\section{Instance Optimality in Weighted Graphs} \label{sec: bid weighted}

Bidirectional Dijkstra, introduced by Dantzig~\cite{original_bid_dijk}
and Nicholson~\cite{Nicholson_orig_bid}, alternates
forward search from $s$ and backward search from $t$
(on reversed edges), combining their distance labels
to certify an $st$-path; see~\cite{bidirectional_history}. 

There are two key properties of Algorithm~\ref{BID_DIJ_AUTH}
\begin{enumerate}
    \item The switch between the forward and backward execution happens on every iteration of the main while loop, so that the exploration from both directions is balanced.
    \item The stopping condition is $d(s,u_s)+d(u_t,t) \ge \mu$ where $\mu$ is the length of the shortest $st$-path computed so far, and $u_s$, $u_t$ are the most recent nodes explored by the forward and backward execution respectively. This was first suggested in \cite{Pohl_BID}.
\end{enumerate}

\begin{algorithm}[H] 
\caption{Original bidirectional Dijkstra implementation (before the correction; see the erratum in \cite{BID_DIJKSTRA_REV2026})} \label{BID_DIJ_AUTH}
\small
\begin{algorithmic}[1] 
\Require Graph $G(V, E)$, source vertex $s$, target vertex $t$
\If{$s=t$} \State \Return (empty path, $0$) \EndIf
\State $\mu \gets +\infty$ \Comment{Length of the shortest path found so far}
\State $e_{\text{mid}} \gets \bot$ \Comment{Middle edge of the shortest path found so far}
\State $u_s \gets s$, $u_t \gets t$ \Comment{Vertices currently being explored in the two executions}
\State Initialize forward search from $s$ on $G$ and backward search from $t$ on $G$ with edges reversed
\While{neither of the two executions has terminated}
    \State Alternate between relaxing one edge in the forward and backward algorithms
\EndWhile
\If{$\mu=+\infty$}
    \State \Return $\bot$ \Comment{One reachable search region is exhausted without an $st$-path}
\EndIf
\State $uv \gets e_{\text{mid}}$
\State $P \gets$ ``shortest $su$-path from forward execution'' $+$ $e_{\text{mid}}$ $+$ ``shortest $vt$-path from backward execution''
\State \Return $(P,\mu)$

\Function{Forward\_Algorithm}{}
    \State $\hat d(s,\cdot) \gets +\infty$; $\hat d(s,s) \gets 0$
    \State Open $s$
    \While{an open vertex exists}
        \State Let $u$ be the open vertex with smallest $\hat d(s,u)$
        \State Close $u$
        \State $u_s \gets u$
        \If{$\hat d(s,u_s) + \hat d(u_t,t) \ge \mu$} \label{line: stoping conditon}
            \State terminate the whole algorithm \label{line: termination}
        \EndIf
        \For{each forward neighbor $v$ of $u$}
            \If{$v$ is not closed}
                \State $\hat d(s,v) \gets \min(\hat d(s,v),\, \hat d(s,u) + \ell(uv))$
            \EndIf
            \If{\textbf{\textcolor{red}{$v$ is closed in the backward execution}}\footnotemark{} \textbf{and} $\hat d(s,u) + \ell(uv) + \hat d(v,t) < \mu$}  \label{line: stop condition}
                \State $\mu \gets \hat d(s,u) + \ell(uv) + \hat d(v,t)$ \label{line: mu update}
                \State $e_{\text{mid}} \gets uv$
            \EndIf
        \EndFor
    \EndWhile
\EndFunction

\Function{Backward\_Algorithm}{}
    \State Analogous to \textsc{Forward\_Algorithm} with roles of forward and backward swapped
\EndFunction

\end{algorithmic}
\end{algorithm}
\footnotetext{The backward-closure requirement highlighted in red
appears in the original implementation of Haeupler et al.
(arXiv versions~1--3 and SOSA 2025). It was subsequently
removed in their revised arXiv version~4~\cite{BID_DIJKSTRA_REV2026},
which acknowledges the independent corrections by Bertram et al.
and the present work.}

When one search exhausts its reachable vertices, the main loop terminates. If $\mu=+\infty$, no $st$-path exists and the algorithm returns $\bot$; if a finite candidate was found, it reconstructs the corresponding path. In particular, no path is reconstructed through an undefined $e_{\mathrm{mid}}$. The same no-path convention applies to the corrected algorithm obtained by deleting the highlighted extra condition.

The additional condition on Line~\ref{line: stop condition},
highlighted in red, prevents the algorithm from updating
$\mu$ unless the corresponding vertex has already been
closed by the opposite search. As we demonstrate below,
this restriction can cause unnecessary exploration and
invalidate the instance-optimality guarantee claimed for
this implementation in the original versions of Haeupler et al.'s
Theorem~3, as documented in the erratum to their revised
paper~\cite{BID_DIJKSTRA_REV2026}.

The difficulty also appears in the original proof of that
theorem. Its indistinguishability argument relies on the
property that, for suitable edges $u_1v_1$ and $v_2u_2$
explored by the forward and backward searches, respectively,
\[
    d(s,u_1)+d(v_2,t)<d(s,t).
\]
However, the additional closure condition can delay the
update of $\mu$, allowing the algorithm to continue
exploring edges even when this inequality no longer holds.
Consequently, the edge-replacement argument used in the
original proof is not justified for the stated implementation.

We illustrate this issue using the graph family $G_D$ and
the competing algorithm $\mathrm{CHEAP}$ introduced in
Section~\ref{sec:Dijkstra}. Throughout this discussion,
the problematic implementation and proof refer specifically
to the original versions of Haeupler et al.'s work,
including the SOSA 2025 proceedings version. Their
subsequent arXiv revision~\cite{BID_DIJKSTRA_REV2026} removes the additional closure
condition and provides a corrected proof of instance
optimality, consistent with the correction studied here.

\begin{claim} \label{BID_not_instance_optimal}
Let $n=3+2D$. The expected number of queries for $\mathrm{CHEAP}$ on graph $G_D$ with inputs $s,t$ is $O(n)$, while the expected number of queries for Algorithm~\ref{BID_DIJ_AUTH} on graph $G_D$ with inputs $s,t$ is $\Omega(n^2)$.
\end{claim}
\begin{proof}
As established in Section~\ref{sec:Dijkstra},
$\mathrm{CHEAP}$ examines only the adjacency lists of
$s$ and $t$, requiring $O(D)=O(n)$ queries.

For the original BID implementation, consider independent
uniformly random tie-breaking among vertices at equal
distance. Since all edges incident to $s$ and $t$ have
weight $1$, the vertex $m$ is indistinguishable from
the $D$ clique neighbors in either search.

With probability at least $1/4$, $m$ appears after at
least $\lfloor D/2\rfloor$ clique vertices in both
tie-breaking orders. On this event, before either search
closes $m$, the algorithm must explore $\Omega(D)$
clique vertices, each having degree $D$. The additional
closure condition prevents the algorithm from discovering
the optimal path and terminating earlier.

Consequently, with constant probability, the original algorithm
performs $\Omega(D^2)$ queries, implying
\[
    T_{\mathrm{original}}(G_D,s,t)=\Omega(D^2)=\Omega(n^2).
\]
In particular, the same lower bound holds on at least one fixed
tie-breaking order, as required for an adversarial-order separation.
\end{proof}

In our constructed example, all edges of the form $(s_i,s_j)$ and $(t_i,t_j)$ for $1\le i<j\le D$ are accessed with probability 0 by $\mathrm{CHEAP}$. If we apply the reasoning from the proof of Theorem 3 in \cite{BID_DIJKSTRA_REV2026}, we can see that the claim is that the shortest $st$-path has its length strictly larger than $d(s,s_i)+d(t_j,t)$. However, in our constructed example $d(s,s_i)+d(t_j,t)=d(s,m)+d(m,t)$ hence $d(s,s_i)+d(t_j,t)$ exactly equals the shortest $st$-path.

\subsection{Instance Optimality of Bidirectional Dijkstra's Algorithm}

In the counterexample above, every pair of edges
$u_1v_1$ and $v_2u_2$ explored by the forward and backward
executions of Algorithm~\ref{BID_DIJ_AUTH}, respectively,
but not queried by $\mathrm{CHEAP}$, satisfies
\[
    d(s,u_1)+d(v_2,t)=d(s,t).
\]
This illustrates why the original implementation does not
satisfy the strict inequality required by its
instance-optimality proof.

The issue can be resolved by a minimal modification of
Algorithm~\ref{BID_DIJ_AUTH}: in the condition on
Line~\ref{line: stop condition}, we remove the requirement
that $v$ is closed in the backward execution, retaining only
\[
    \hat d(s,u)+\ell(uv)+\hat d(v,t)<\mu.
\]
We denote the resulting corrected algorithm by \BID{}
(BIdirectional Dijkstra). This modification coincides with
the correction subsequently incorporated into the revised
version of~\cite{BID_DIJKSTRA_REV2026}.

We now provide an independent, self-contained proof of
instance optimality of \BID{} on positively weighted
multigraphs. The key property is that, for every pair
of edges $u_1v_1$ and $v_2u_2$ explored by the forward
and backward searches, respectively, we have
\[
    d(s,u_1)+d(v_2,t)<d(s,t).
\]
We establish this property through the following lemmas.
The revised proof of Haeupler et al.~\cite{BID_DIJKSTRA_REV2026}
also establishes and exploits this inequality in its
instance-optimality argument, although it derives it
through a different auxiliary stopping lemma.

\begin{lemma} \label{m closed}
    Consider \BID on an input $(G,s,t)$ with finite shortest-path distance $\mu=d(s,t)$. Suppose, without loss of generality, that the final update setting $\mu$ to its optimal value is made by the forward execution through an edge $(u^s_\mu,m)$, as
    $\mu\gets d(s,u^s_\mu)+\ell(u^s_\mu,m)+\hat d(m,t)$.
    Then, by the time $m$ has been closed and fully explored by the backward execution, $\mu$ must already have its final value.
\end{lemma}
\begin{proof}
See Appendix~\ref{proof of m closed}.
\end{proof}

\begin{lemma} \label{explored edge distance}
Consider \BID on some input instance $(G,s,t)$. Let \(E_s\) and \(E_t\) be the sets of edges explored
by the forward and backward executions, respectively, and let \(\mu=d(s,t)\). Then for every edge \(u_1v_1\in E_s\)
explored from \(u_1\), and every edge \(v_2u_2\in E_t\) explored from
\(v_2\), we have
\[
    d(s,u_1)+d(v_2,t)<\mu.
\]
\end{lemma}

\begin{proof}
If $d(s,t)=+\infty$, every explored forward issuing vertex is
reachable from $s$ and every explored backward issuing vertex can reach
$t$, so the inequality is immediate. Assume henceforth that
$d(s,t)$ is finite.
For the sake of contradiction, assume there exists a pair of explored edges $u_1v_1, v_2u_2$, explored from the forward and backward executions respectively, from $u_1$ and $v_2$ respectively, such that $d(s,u_1)+d(v_2,t)\ge \mu$ where $\mu$ is the length of the shortest $st$-path. It follows that both $u_1$ and $v_2$ are closed in their respective executions. 
    
To demonstrate a contradiction it suffices to show that when the last of $u_1$ or $v_2$ is closed the final update of $\mu$ is already done, since then the algorithm would terminate on Line~\ref{line: termination} without having a chance to explore the given edge. Without loss of generality assume the final update of $\mu$ was made from the forward execution as $\mu \gets d(s,u^s_\mu)+\ell(u^s_\mu, m)+\hat d(m,t)$. The assumption is therefore $d(s,u_1)+d(v_2,t)\ge \mu = d(s,u^s_\mu)+\ell(u^s_\mu, m)+\hat d(m,t)$. 

First, consider the case when $d(s,u_1) > d(s,u^s_\mu)$. In that case $u^s_\mu$ is already closed and explored in the forward execution when $u_1$ gets closed and hence $\mu$ already has its final value. In that case whichever node is closed second will trigger the termination of the algorithm.

Now consider the second case $d(s,u_1) \le d(s,u^s_\mu)$. This implies that $d(v_2,t) \ge \ell (u^s_\mu,m)+\hat d(m,t)$, which in turn implies that $m$ is already closed and explored at the time that $v_2$ is closed in the backward execution. By Lemma \ref{m closed} it follows that if $m$ is closed and explored the $\mu \gets d(s,u^s_\mu)+\ell(u^s_\mu, m)+\hat d(m,t)$ update had already been done. As before, whichever node is closed second will trigger the termination of the algorithm. 

This shows that it is impossible for such a pair of edges $u_1v_1$ and $v_2u_2$ to exist, which concludes our proof.
\end{proof}

We now use Lemma~\ref{explored edge distance} to give an
independent proof of the instance-optimality guarantee
originally claimed by Haeupler et al.~
\cite[Theorem 3]{BID_DIJKSTRA_REV2026}, for the corrected
implementation of \BID{}.

\begin{theorem}[Instance optimality of corrected \BID{}]
\label{omission_optimality}
    \BID{} is an instance-optimal algorithm, under query
    complexity, for the shortest $st$-path problem in both
    directed and undirected multigraphs with positive
    edge weights.
\end{theorem}

\begin{proof} 
See Appendix~\ref{proof of omission_optimality}.
\end{proof}

\section{Instance Optimality in Simple Graphs} \label{sec:Simple Graphs}

The edge modifications used in the multigraph proof may
create parallel edges, so they do not directly establish
instance optimality for simple graphs. We first study
a local-discovery model in which the modification remains
possible, then turn to structural certificates in the
standard fixed-universe model.

\subsection{Local discovery and an extensible vertex universe}
\label{subsec:local-discovery}

The distinction between multigraphs and simple graphs is important in the
usual adjacency-list query model, where the vertex set is fixed and its
cardinality is known to the algorithm. In particular, a degree query can then
certify that a vertex is adjacent to every other vertex, without revealing
those incident edges individually. This prevents the degree-preserving edge
switch used in the multigraph proof whenever one of the desired replacement
edges already exists.

There is, however, another natural access model in which the obstruction
vanishes. Many graph exploration algorithms initially know only the source
and target and discover other vertices by following adjacency-list entries.
For these algorithms, adding an as-yet-undiscovered vertex need not be
observable. The relevant restriction is stronger than merely requiring an
algorithm not to refer explicitly to $|V(G)|$: the algorithm must also be
unable to query arbitrary undiscovered vertex identifiers or obtain other
global information about the vertex universe.

\begin{definition}[Local-discovery query model]
\label{def:local-discovery}
The input is a finite graph with distinguished vertices $s$ and $t$, whose
identifiers are initially available to the algorithm. Vertex identifiers
are treated as opaque handles. The algorithm may issue degree queries and
indexed adjacency-list queries, but only at vertices whose handles it has
already obtained. In the directed case, it may query incoming and outgoing
adjacency lists and their respective indegrees and outdegrees, each
at unit cost. An outgoing-list query reveals the head of the arc,
whereas an incoming-list query reveals its tail; both reveal the arc's
weight. In the undirected case, an adjacency query reveals the
other endpoint and weight of the edge. The algorithm learns
new handles only through such queries. In particular, it is not given the
number of vertices, an enumeration of their handles, or another oracle for
accessing undiscovered vertices. Its next action and output depend only on
$s,t$, its internal randomness, and its observed query transcript.
Correctness, including the $\bot$ output for disconnected terminals, is required over finite graphs of arbitrary sizes.
\end{definition}

Unlike the outgoing-only oracle in
Haeupler et al.'s Theorem~2, this model permits
degree queries and, in directed graphs, incoming-edge
exploration from $t$. Theorem~\ref{thm:bid-local-discovery}
concerns bidirectional search: its proof uses the fact
that inserting an undiscovered vertex changes no
previously observed handle or query answer.
Corrected \BID{} maintains only discovered vertices
and returns $\bot$ when one reachable search
region is exhausted without an $st$-path.

\begin{theorem}[Instance optimality under local discovery]
\label{thm:bid-local-discovery}
On the class of finite \emph{directed and undirected} simple graphs
(with neither loops nor parallel edges) and strictly positive edge
weights, the corrected \BID{} is instance
optimal in the local-discovery query model against randomized algorithms
that are correct with probability at least $0.9$ on every admissible input.
For distinct terminals $s,t$, every such algorithm $A$, and
\emph{every fixed adjacency-list ordering},
\[
T_{\BID}(G,s,t)\le 30\,T_A(G,s,t).
\]
In particular, this guarantee holds in both the order-dependent and
order-oblivious interpretations of local discovery, independently of
directedness, provided incoming adjacency queries are available
in the directed case.
\end{theorem}

\begin{proof}
See Appendix~\ref{proof of thm:bid-local-discovery}.
\end{proof}

The guarantee is pointwise for every handle assignment
and adjacency ordering; no random representation is
needed. It does not extend to a fixed, known vertex
universe, where algorithms may access undiscovered
vertex identifiers or detect the introduction of a
new vertex from the vertex count. Local discovery
also excludes global random-vertex sampling and
preprocessed landmark or distance-oracle methods
that rely on such access
\cite{GoldbergHarrelson2005,ThorupZwick2005,GoldreichRon2002}.

An interesting intermediate setting would permit uniform
random-vertex sampling without revealing the vertex count
or allowing arbitrary indexed access. In this setting,
the indistinguishability argument may still be effective.
Indeed, the construction introduces only one new vertex.
If the original graph contains $n$ vertices, each uniform
sample from the modified graph encounters this vertex
with probability $1/(n+1)$. Under a coupling of
the sampling procedures, if an algorithm makes $R$ such
samples on the original instance, the probability of
distinguishing the modification through random sampling
is at most
\[
    \frac{\mathbb{E}[R]}{n+1}.
\]
Thus, when $\mathbb{E}[R]=o(n)$, the additional probability
of distinguishing the two instances through sampling
vanishes asymptotically.

This observation suggests that the proof technique could
extend to certain randomized global-access procedures.
Nevertheless, it does not establish instance optimality
in the more permissive model: an algorithm making
$\Theta(n)$ random samples may encounter the inserted
vertex with non-negligible probability, and knowledge
of the vertex count would reveal the modification
immediately. Establishing guarantees under such additional
query operations therefore requires a separate analysis.

\subsection{Instance-sensitive bounds with a fixed vertex universe}
\label{subsec:fixed-universe-simple}

We return to the standard fixed-universe model, in
which arbitrary vertex identifiers may be queried.
The next results give execution-sensitive certificates
against universally correct algorithms with such access.
The reservoir argument uses vertices with low access
probability; the tensor argument uses pairs of low-query
forward and backward first-exploration edges.
The latter also yields refinements using explored cores,
degree distributions, and a necessary concentration
condition for superconstant speedups.

Throughout this subsection the graphs are undirected,
simple, and strictly positive-weighted. The
degree-preserving modifications used below need not
preserve directed in- and outdegrees.

\begin{definition}
    A node is considered accessed, if its degree or an edge incident to it is queried.
\end{definition}

\begin{theorem}\label{thm:reservoir_node}
Consider \BID{} on an undirected simple positive-weight
instance $(G,s,t)$ with distinct terminals $s\ne t$,
and let $A$ be any randomized algorithm that is correct
with probability at least $0.9$ on every input.
Suppose there exists a vertex $r\in V(G)$ that $A$ accesses with probability at most $0.25$ on $G$. Then, for a universal constant $\alpha>0$,
\[
T_A(G,s,t)>\alpha S.
\]
In particular, one may take $\alpha=1/32$.
\end{theorem}

\begin{proof}
See Appendix~\ref{proof of thm:reservoir_node}.
\end{proof}

\begin{corollary}[Instance-sensitive reservoir bound]
\label{cor:reservoir_all_algorithms}
Let $(G,s,t)$ be an undirected simple graph with strictly
positive edge weights and distinct terminals $s\ne t$.
Consider the execution of \BID. Then every randomized
algorithm $A$ that returns a correct
shortest $st$-path and its length with probability at least
$0.9$ on every admissible input satisfies
\[
\boxed{
T_A(G,s,t)\ge
\min\left\{\frac{S}{32},\frac{n}{8}\right\}.
}
\]
Moreover, for every such algorithm $A$,
\[
T_{\BID}(G,s,t)\le
\max\left\{128,\frac{16(S+1)}{n}\right\}
T_A(G,s,t).
\]
In particular, the instance-optimality ratio of \BID is
$O(1+S/n)$.
\end{corollary}

\begin{proof}
Fix an arbitrary algorithm $A$ satisfying the stated
correctness requirement.
If some vertex $r$ is accessed by $A$ with probability
at most $1/4$, Theorem~\ref{thm:reservoir_node}
gives $Q>S/32$.

Otherwise, every vertex is accessed with probability
greater than $1/4$. Since each query accesses at most
two vertices,
\[
2Q\ge
\sum_{v\in V(G)}
\Pr[A\text{ accesses }v]
>\frac{n}{4},
\]
so $Q>n/8$. In either case,
\[
Q\ge\min\left\{\frac{S}{32},\frac{n}{8}\right\}.
\]

The two searches alternate edge explorations; hence
\BID makes $S$ indexed-neighbor queries and at most
$S+2$ degree queries. Therefore,
\[
T_{\BID}(G,s,t)\le 2S+2.
\]
For $S\ge1$, dividing by the preceding lower bound
gives
\[
\begin{aligned}
\frac{T_{\BID}(G,s,t)}{T_A(G,s,t)}
&\le
\frac{2S+2}{\min\{S/32,n/8\}}\\
&=\max\left\{
64\left(1+\frac1S\right),
\frac{16(S+1)}{n}
\right\}\\
&\le\max\left\{128,\frac{16(S+1)}{n}\right\}.
\end{aligned}
\]

If $S=0$, \BID performs at most two degree queries.
Because $s\ne t$ and the admissible inputs have
unrestricted positive edge weights, any algorithm
that is correct with probability at least $0.9$
on every admissible input must make a query with
probability at least $0.8$. Thus $Q\ge0.8$, and
the same ratio bound holds.
\end{proof}

\begin{theorem}[Sharper quantitative reservoir bound]
\label{thm:quantitative_reservoir}
Let $(G,s,t)$ be an undirected simple graph with strictly positive
edge weights and distinct terminals $s\neq t$, and consider the
execution of \BID. In this theorem, abbreviate
$a=|E_s|$ and $b=|E_t|$. If $a,b\ge1$,
then every randomized algorithm $A$ that is correct
with probability at least $0.9$ on every input satisfies
\[
T_A(G,s,t)\ge
\frac{2/5}{1/a+1/b+1/n}.
\]
Consequently, using the globally defined optimal
competitor complexity $\tau(G,s,t)$, we have
\[
\frac{T_{\BID}(G,s,t)}{\tau(G,s,t)}
\le 30+\frac{5(S+1)}{n}.
\]
The latter inequality also holds when $a=0$ or $b=0$.
\end{theorem}

\begin{proof}
See Appendix~\ref{proof of thm:quantitative_reservoir}.
\end{proof}

\begin{corollary}\label{cor:large_vertex_reservoir}
Consider a class of undirected simple positive-weight
instances $(G,s,t)$ for which the execution of \BID satisfies
\[
|V(G)|>\frac{|E_s|+|E_t|}{16}.
\]
Then
\[
\frac{T_{\BID}(G,s,t)}{\tau(G,s,t)}
<110+\frac{5}{|V(G)|}.
\]
In particular, \BID is instance optimal up to a
constant factor on this class.
\end{corollary}

\begin{proof}
The result follows immediately from
Theorem~\ref{thm:quantitative_reservoir}, since
$S=|E_s|+|E_t|<16|V(G)|$.
\end{proof}

\begin{corollary}[Sparse-side executions]
\label{thm:sparse_side_randomized}
Under the hypotheses of
Theorem~\ref{thm:quantitative_reservoir}, suppose that
\[
|E_s|\le K|N_s|
\]
for some constant $K\ge0$.

Then
\[
\frac{T_{\BID}(G,s,t)}{\tau(G,s,t)}
\le30+10K+\frac{10}{n}.
\]

In particular, \BID is instance optimal up to a constant
factor on every class of instances with a fixed $K$.
\end{corollary}

\begin{proof}
Since $|E_t|\le|E_s|+1$ and $|N_s|\le n$, we have 
\(
S\le2|E_s|+1\le2Kn+1.
\)
Substituting into
Theorem~\ref{thm:quantitative_reservoir} gives
\[
\frac{T_{\BID}(G,s,t)}{\tau(G,s,t)}
\le30+\frac{5(S+1)}n\
\le30+10K+\frac{10}{n}.
\]
\end{proof}

\begin{corollary}[Explicit graph-density bound]
\label{cor:explicit_graph_density}
For every undirected simple graph $G$ with strictly positive
edge weights,
\[
\frac{T_{\BID}(G,s,t)}{\tau(G,s,t)}
\le
30+20\frac{|E(G)|}{|V(G)|}
+\frac5{|V(G)|}.
\]
\end{corollary}

\begin{proof}
Every undirected edge can be explored at most once from each
endpoint in each of the two searches. Consequently,
\(
S\le4|E(G)|.
\)
The result follows immediately from
Theorem~\ref{thm:quantitative_reservoir}.
\end{proof}

\paragraph{A weight-sensitive certificate.}
Unlike the preceding structural bounds, the following
certificate applies to directed and undirected graphs,
with or without parallel edges, in both the standard
and local-discovery query models. It also holds for
every fixed adjacency ordering and therefore under
order-oblivious averaging. The only essential
assumptions are that edge weights may be lowered
independently to arbitrary positive values within
the admissible input class, that individual edge
weights are revealed by adjacency queries, and
that the required output determines the shortest-path
distance.

Write $d_s(v)=d(s,v)$, $d_t(v)=d(v,t)$ and
$\mu=d(s,t)$. An underlying edge or arc is
\emph{weight-critical} if reducing only its weight
to a sufficiently small positive value strictly
decreases $\mu$. The set of all such edges or arcs is
\[
\mathcal C_{s,t}=
\begin{cases}
\{(u,v)\in E(G):d_s(u)+d_t(v)<\mu\},
 & \text{directed }G,\\[2pt]
\bigl\{\{u,v\}\in E(G):
 \min\{d_s(u)+d_t(v),d_s(v)+d_t(u)\}<\mu\bigr\},
 & \text{undirected }G.
\end{cases}
\]
Parallel edges or arcs are counted separately.
For $\mu=+\infty$, the set is empty:
decreasing a positive edge weight cannot change
reachability.

\begin{proposition}[Weight-critical-edge certificate]
\label{prop:critical-weight-edges}
Let $(G,s,t)$ be a directed or undirected
positive-weight graph, simple or a multigraph,
with distinct terminals. Consider any adjacency-list
query model in which weights can be independently
reduced while preserving the admissible input class,
and which reveals individual edge weights only
through queries of the corresponding edge or arc.
Every randomized algorithm $A$ that outputs the
correct shortest-path distance with probability
at least $0.9$ on every admissible input satisfies
\[
p_A(e)\ge\frac45\quad(e\in\mathcal C_{s,t}),
\qquad
T_A(G,s,t)\ge\frac45|\mathcal C_{s,t}|.
\]
The bound holds for every fixed representation,
and hence also for randomized adjacency orderings,
independently of whether arbitrary vertex access
or only local discovery is allowed. In any of these
models that supports the forward and backward
explorations of corrected \BID{}, the condition
$|\mathcal C_{s,t}|=\Omega(S)$ on a family of
instances implies its instance optimality.
\end{proposition}
\begin{proof}
See Appendix~\ref{proof of prop:critical-weight-edges}.
\end{proof}

Every weight-critical edge or arc is explored
by corrected \BID{}, so $|\mathcal C_{s,t}|\le S$.
For finite $\mu$, let $R_s,R_t$ be the last
closed distances when the distance-bound
stopping condition fires, with
$R_s+R_t\ge\mu$. If $e$ is critical in
orientation $u\to v$, then
$d(s,u)+d(v,t)<\mu$, so either
$d(s,u)<R_s$ or $d(v,t)<R_t$.
Accordingly, the forward search has
fully explored $u$, or the backward
search has fully explored $v$, revealing
$e$ through the appropriate outgoing
or incoming incidence entry. If instead
one search exhausts its reachable region,
every critical edge is revealed by that
search. The set is empty for $\mu=+\infty$.
There is no converse upper bound on
$S$ in terms of $|\mathcal C_{s,t}|$ alone.

The weight-critical-edge criterion can succeed even when the
reservoir, bounded-degree-core, degree-moment, orientation,
degree-profile, and exceptional-degree criteria in
Table~\ref{tab:additional-regimes} do not prove a constant
ratio. Appendix~\ref{app:critical-edge-separation} gives an
explicit undirected separating family. The same appendix
also shows that the converse fails dramatically:
$|\mathcal C_{s,t}|$ can remain constant while
$S=\Theta(n^2)$.

\begin{lemma}\label{lem:at_most_one_boundary_pair}
Let $G=(V,E)$ be a positive-weight multigraph, and consider corrected \BID{} on input $(G,s,t)$. Let $\mu$ denote the length of the shortest $st$-path. Suppose vertices $p_s,p_t$ are closed by the forward and backward
searches, respectively, and satisfy
\[
d(s,p_s)+d(p_t,t)\ge \mu .
\]
Then the algorithm terminates when the later of $p_s$ and $p_t$ is closed,
before exploring any edge from that vertex.

In particular, if two such pairs $(p_s^1,p_t^1)$ and $(p_s^2,p_t^2)$ satisfy
\[
d(s,p_s^1)+d(p_t^1,t)\ge \mu
\qquad\text{and}\qquad
d(s,p_s^2)+d(p_t^2,t)\ge \mu,
\]
then either $p_s^1=p_s^2$ or $p_t^1=p_t^2$.
\end{lemma}

\begin{proof}
See Appendix~\ref{proof of lem:at_most_one_boundary_pair}.
\end{proof}

\begin{lemma}\label{lem: cross edge queries}
    Consider executing \BID on some simple graph $G$ with positive edge weights, on inputs $s,t$. Let $A$ be any correct algorithm for the shortest $st$-path problem. Let $G_s$ and $G_t$ be the low-query first-exploration subgraphs, formed from $E_s^{(1)}$ and $E_t^{(1)}$. More precisely, define
    \[
        E(G_s)=\{e\in E_s^{(1)}:\Pr[A\text{ queries }e]\le1/4\}, \quad E(G_t)=\{e\in E_t^{(1)}:\Pr[A\text{ queries }e]\le1/4\}
    \]
    and
    \[
        V(G_s)=\{v: \exists e\in E(G_s)\text{ incident to }v\}, \quad V(G_t)=\{v: \exists e\in E(G_t)\text{ incident to }v\}.
    \]
    Let $X=V(G_s)$, $Y=V(G_t)$, and let $m=e(X,Y)$ be the number of cross-edges with one endpoint in $X$ and the other in $Y$. Then for every edge $e\in E(X,Y)$, $A$ must query $e$ with probability at least $0.25$. Consequently $T_A(G,s,t)\ge m/4$.
\end{lemma}
\begin{proof}
Suppose, for a contradiction, that an edge $xy\in E(X,Y)$ is
queried by $A$ with probability less than $1/4$. Choose an incidence
record $e_s\in E(G_s)$ incident to $x$, and a record
$e_t\in E(G_t)$ incident to $y$. Let $x'$ and $y'$ be the vertices
from which these records were queried by the forward and backward
searches, respectively. Thus $x'=x$ or $e_s=x'x$, and $y'=y$ or
$e_t=y'y$. By Lemma~\ref{explored edge distance},
\[
 d(s,x')+d(y',t)<\mu:=d(s,t).
\]
Lower the weights of $xy$ and, when nontrivial, $x'x$ and $yy'$,
so that their total weight is less than
$\mu-d(s,x')-d(y',t)$. No edge is removed, so the original
shortest $sx'$- and $y't$-paths remain available (possibly with
smaller weights). The modified graph consequently contains an
$st$-walk, and hence a path, of length strictly less than $\mu$.
Only the weights of $xy,e_s,e_t$ have changed. Since each is
queried by $A$ with probability at most $1/4$, the union bound
shows that the coupled executions on the original and modified
instances have identical transcripts with probability at least $1/4$.
Their shortest-path distances differ, so at least one execution
has error probability at least $1/8>0.1$, a contradiction.
Thus every cross edge is queried with probability at least $1/4$,
and summing these probabilities gives $T_A(G,s,t)\ge m/4$.
\end{proof}

\begin{definition}[Open and closed nodes in explored edges] \label{def: open/closed nodes in edges}
    Consider executing \BID on some graph $G$ with positive edge weights, on inputs $s,t$. Let $uv\in E(G)$ be an edge queried by the algorithm. Assume the edge $uv$ was explored from the closed node $u$ as $u \rightarrow v$. We say that the node $u$ is closed and node $v$ is open with respect to the queried edge $uv$.
\end{definition}

Note that the notion of open and closed nodes as per Definition~\ref{def: open/closed nodes in edges} is not the same as open and closed nodes by Algorithm~\ref{BID_DIJ_AUTH} in general. A node might be closed by the algorithm, however, with respect to the edge that closed it, we would consider it to be open.

\begin{lemma}[Required cross edges for first explorations]
\label{lem: required cross edges}
Let $G_s,G_t$ be the low-query first-exploration subgraphs
of the forward and backward searches. For every forward
record $x\to x'\in E(G_s)$ and backward record
$y\to y'\in E(G_t)$, at least one of $xy,x'y'$ belongs to $E(G)$.
\end{lemma}
\begin{proof}
See Appendix~\ref{proof of lem: required cross edges}.
\end{proof}

\begin{lemma}\label{lem:low_query_intersection}
Consider the setup of Lemma~\ref{lem: required cross edges}. Then
\[
V(G_s)\cap V(G_t)=\emptyset.
\]
Consequently,
\[
E(G_s)\cap E(G_t)=\emptyset.
\]
\end{lemma}

\begin{proof}
See Appendix~\ref{proof of lem:low_query_intersection}.
\end{proof}

\begin{theorem}[Universal tensor lower bound]\label{thm:tensor_cross_edges}
Fix an undirected simple graph with strictly positive edge weights
and \emph{any} randomized algorithm $A$ that returns a correct
shortest $st$-path and its length with probability at least $0.9$
on \emph{every admissible input}. Let $G_s,G_t$ be BID's
low-query first-exploration subgraphs formed relative to $A$,
as in Lemma~\ref{lem: required cross edges}. Let \(X=V(G_s)\text{, }Y=V(G_t),\)
and write \(q=e(X,Y)\text{, } M_s=e(G_s)\text{, } M_t=e(G_t).\)
Let
\[
\Delta_s=\max_{x\in X} d_{G_s}(x),
\qquad
\Delta_t=\max_{y\in Y} d_{G_t}(y).
\]
Assume $M_s,M_t>0$ (if either is zero, the tensor method
makes only the trivial assertion $Q\ge0$).
Then
\[
q\ge \frac{M_sM_t}{\Delta_s\Delta_t}.
\]
Moreover,
\[
Q\ge\frac{q}{4}\ge\frac{M_sM_t}{4\Delta_s\Delta_t}.
\]
These lower bounds hold for every such algorithm $A$;
$G_s,G_t$ may depend on $A$, but no aspect of $A$'s
query strategy is restricted.
Consequently, if \(M_s,M_t=\Omega(M)\) for some parameter \(M\), and \(\Delta_s\Delta_t=O(M),\) then \(q=e(X,Y)=\Omega(M).\)
In particular, if \(\Delta_s,\Delta_t=O(\sqrt M),\) then \(q=\Omega(M).\)
By Lemma~\ref{lem: cross edge queries}, this implies \(T_A(G,s,t)=\Omega(M).\)
\end{theorem}

\begin{proof}
Each explored incidence record in $G_s$ or $G_t$ carries its actual exploration orientation: the tail is the vertex whose adjacency entry was queried, and the head is its neighbor. Thus the tail is closed with respect to that record and the head is open, as in Definition~\ref{def: open/closed nodes in edges}. If an underlying edge is later queried from the opposite endpoint within the same search, the reverse incidence record is not retained in this tensor construction; it still counts toward BID's total query complexity.

We define an auxiliary graph \(H\) with vertex set \(V(H)=X\times Y.\)
For every oriented edge \(x\to x'\) in \(G_s\),
and every oriented edge \(y\to y'\) in \(G_t\),
put an edge in \(H\) between the vertices \((x,y)\) and \((x',y')\), retaining multiplicity if different record pairs produce the same auxiliary endpoints. Note that by Lemma~\ref{lem:low_query_intersection} we have that for every $u\in X\cup Y$ it follows $(u,u)\not\in V(H)$.

We claim that the set of missing cross edges
\[
\overline{E}(X,Y)
=
(X\times Y)\setminus E(X,Y)
\]
is an independent set in \(H\). Indeed, suppose not. Then there exist oriented edges  \(x\to x'\) in \(G_s\) and \(y\to y'\) in \(G_t\)
such that both \(xy\notin E(X,Y)\) and \(x'y'\notin E(X,Y)\).
But \(x\) is closed with respect to the edge \(xx'\), and \(y\) is closed with respect to the edge \(yy'\). Therefore, by Lemma~\ref{lem: required cross edges}, the missing cross edge \(xy\) forces
\(
x'y'\in E(X,Y),
\)
a contradiction.

Hence the missing cross edges form an independent set in \(H\). Therefore the present cross edges \(E(X,Y)\) form a vertex cover of \(H\). Thus
\(
q=e(X,Y)\ge \tau(H).
\)

Now we estimate \(\tau(H)\). The auxiliary multigraph \(H\) has one edge, counted with multiplicity, for every pair consisting of one incidence record of \(G_s\) and one of \(G_t\). Hence \(e(H)=M_sM_t.\)
Moreover, for every \((x,y)\in X\times Y\),
\[
d_H(x,y)
\le d_{G_s}(x)d_{G_t}(y)
\le \Delta_s\Delta_t.
\]
Therefore \(\Delta(H)\le \Delta_s\Delta_t.\)
Using the elementary vertex-cover bound
\[
\tau(H)\ge \frac{e(H)}{\Delta(H)},
\]
we get
\[
q
\ge \tau(H)
\ge
\frac{e(H)}{\Delta(H)}
\ge
\frac{M_sM_t}{\Delta_s\Delta_t}.
\]

This proves the claimed tensor lower bound.

If \(M_s,M_t=\Omega(M)\) and \(\Delta_s\Delta_t=O(M)\), then
\[
q\ge
\frac{\Omega(M)\Omega(M)}{O(M)}
=
\Omega(M).
\]
In particular, \(\Delta_s,\Delta_t=O(\sqrt M)\) implies \(\Delta_s\Delta_t=O(M)\), so again \(q=\Omega(M).\)
Finally, Lemma~\ref{lem: cross edge queries} states that every edge in \(E(X,Y)\) must be queried by \(A\) with probability at least \(0.25\), and therefore
\[
T_A(G,s,t)=\Omega(q)=\Omega(M).
\]
\end{proof}

\begin{corollary}\label{cor:degree_bounded_by_bid_search}
Consider a class of undirected simple positive-weight input instances $\mathcal I$ for the shortest $st$-path problem such that for every $(G,s,t)\in \mathcal I$, the execution of Algorithm~\BID satisfies
\[
\sqrt{|E_s|+|E_t|}\ge \Delta,
\]
where $E_s$ and $E_t$ are the sets of edges explored by the forward and backward executions, respectively, and $\Delta$ is the maximum degree of $G$. Then every randomized algorithm $A$ that is correct with probability at least $0.9$ on every input satisfies
\[
T_A(G,s,t)
\ge
\frac{1}{1024}\bigl(|E_s|+|E_t|\bigr)
\]
on every instance $(G,s,t)\in\mathcal I$. Consequently, Algorithm~\BID is instance optimal up to a constant factor on $\mathcal I$.
\end{corollary}

\begin{proof}
See Appendix~\ref{proof of cor:degree_bounded_by_bid_search}
\end{proof}

Note that a more general version of Corollary~\ref{cor:degree_bounded_by_bid_search} can be obtained if one does not consider the maximum degree in the entire graph, rather the maximum degree in the two subgraphs induced by the explored edges and nodes of the two executions in \BID{}. 

\subsection{Refined tensor bounds from explored-edge distributions}\label{sec:refined-tensor}

The maximum-degree estimate in Theorem~\ref{thm:tensor_cross_edges} is convenient, but may be dominated by a small number of high-degree vertices. We develop refinements that retain the original proof's combinatorial mechanism: for every correct algorithm \(A\), the explored edges that \(A\) queries with probability at most \(1/4\) define an auxiliary tensor graph \(H\), in which the vertices corresponding to actual cross edges form a vertex cover. By Lemma~\ref{lem: cross edge queries}, each of those cross edges contributes at least $1/4$ to the expected query count. Throughout this subsection the graph is undirected and simple, with positive edge weights, and we use the conventions of Section~\ref{sec:Preliminaries}. All asymptotic constants in the sufficient conditions below are uniform over the stated class of instances. Since the searches alternate edge explorations, $|E_s|,|E_t|=\Theta(S)$ outside finitely many trivial cases.

\begin{theorem}[Arbitrary explored cores]\label{thm:explored-cores}
Let $F_s\subseteq E_s^{(1)}$ and $F_t\subseteq E_t^{(1)}$ be any two nonempty edge sets fixed independently of the competing algorithm $A$. Set
\[
m_s=|F_s|,\quad m_t=|F_t|,\quad
D_s=\Delta(F_s),\quad D_t=\Delta(F_t),
\]
where \(\Delta(F)\) denotes the maximum number of edge queries in \(F\) incident to any single vertex. Only the first exploration of each underlying edge is retained in a chosen core, while reverse-incidence queries count toward BID's full execution cost. Every randomized algorithm $A$ that returns a correct shortest $st$-path and its length with probability at least $0.9$ on every admissible input satisfies
\[
\boxed{\quad Q\ge\frac1{16}\min\left\{m_s,m_t,
                     \frac{m_sm_t}{D_sD_t}\right\}.\quad}
\]
\end{theorem}
\begin{proof}
See Appendix~\ref{proof of thm:explored-cores}
\end{proof}

\begin{corollary}[Bounded-degree explored cores]\label{cor:bounded-explored-cores}
Consider a family of instances for which there exist $F_s\subseteq E_s^{(1)}$ and $F_t\subseteq E_t^{(1)}$ satisfying
\[
|F_s|,|F_t|=\Omega(S),\qquad
\Delta(F_s)\Delta(F_t)=O(S).
\]
Then the corrected \BID{} is instance optimal on this family.
\end{corollary}
\begin{proof}
Theorem~\ref{thm:explored-cores} implies $Q=\Omega(S)$ for every correct $A$. Since \BID{} makes $O(S)$ queries, the claim follows.
\end{proof}

\begin{corollary}[Few exceptional high-degree vertices]
\label{cor:few-exceptional-degrees}
Let $h_D$ be the number of vertices of an undirected
simple positive-weight graph $G$ whose degree exceeds
an integer $D\ge1$. In the fixed-vertex-universe model,
\[
\boxed{\quad
\rho_{\mathrm{BID}}(G,s,t)
=O\left(1+h_D+\frac{D^2}{n}\right).
\quad}
\]
In particular, if $h_D=O(1)$ and $D=O(\sqrt n)$ along
a family of instances, corrected \BID{} is instance
optimal even if the exceptional vertices have degree
$\Theta(n)$.
\end{corollary}
\begin{proof}
See Appendix~\ref{proof of cor:few-exceptional-degrees}.
\end{proof}

Bertram et al.~\cite[Theorem~8.15]{BertramJensenThorupWangYan2026}
prove a bound of the same functional form for directed graphs.
That directed statement does not by itself imply our
undirected result by encoding every undirected edge as
a pair of opposite arcs. Although the encoded graphs
form a subclass of directed graphs, an algorithm
promised an undirected input need only be correct on
that subclass. It can therefore exploit the symmetry
of the input in ways unavailable to an algorithm
required to work on \emph{all} directed graphs.
The directed competitor class is smaller, so a
lower bound against that class need not hold against
every undirected-only competitor.

\begin{proposition}[An instance family with exceptional hubs]
\label{prop:clique-hub-separation}
For every fixed $\alpha\in(3/2,2)$, there exists a family of
connected undirected positive-weight instances satisfying
\[
n=\Theta(r^\alpha),\qquad
S=\Theta(r^2),\qquad
\Delta=\Theta(r^\alpha),
\]
for which the bounded-degree-core condition holds, whereas none
of the following sufficient conditions holds:
\[
\Delta^2=O(S),\qquad
P_sP_t=O(S^3),\qquad
S=O(n).
\]
In particular, the bounded-degree-core condition establishes
instance optimality on this family, while the untrimmed
maximum-degree, second-moment, and reservoir bounds do not
establish a constant-factor guarantee.
\end{proposition}

\begin{proof}
See Appendix~\ref{app:exceptional-hubs-proof}.
\end{proof}

The freedom to choose $F_s,F_t$ is essential: a small number of exceptional high-degree vertices may cause $\Delta(E_s)\Delta(E_t)\gg S$ even if a constant fraction of each explored edge set lies in a much less concentrated subgraph. It also admits a quantitative version. Define the \emph{core parameter}
\[
\mathcal L(E_s,E_t)=
\max_{\substack{\varnothing\ne F_s\subseteq E_s^{(1)}\\
                  \varnothing\ne F_t\subseteq E_t^{(1)}}}
\min\left\{|F_s|,|F_t|,
\frac{|F_s||F_t|}{\Delta(F_s)\Delta(F_t)}\right\}.
\]
By Theorem~\ref{thm:explored-cores},
\begin{equation}\label{eq:core-ratio}
\tau(G,s,t)\ge \frac{\mathcal L(E_s,E_t)}{16},
\qquad
\rho_{\mathrm{BID}}(G,s,t)=O\!\left(\frac{S}{\mathcal L(E_s,E_t)}\right),
\end{equation}
This is an instance-wise bound, not a global $n$- or $m$-parameterized worst-case bound.

\begin{proposition}[Computing the optimal core parameter]\label{prop:bmatching-cores}
For an edge set $F$, let
\[
b_F(d)=\max\{|J|:J\subseteq F,\ \Delta(J)\le d\}.
\]
Then
\[
\mathcal L(E_s,E_t)=\max_{d_s,d_t\ge1}
\min\left\{b_{E_s^{(1)}}(d_s),b_{E_t^{(1)}}(d_t),
 \frac{b_{E_s^{(1)}}(d_s)b_{E_t^{(1)}}(d_t)}{d_sd_t}\right\},
\]
where it suffices to consider integer $d_s\le\Delta(E_s^{(1)})$ and $d_t\le\Delta(E_t^{(1)})$. For each fixed degree bound $D$, computing $b_s(D)$
is a maximum-cardinality $b$-matching problem with
uniform vertex capacities $b(v)=D$, which is solvable
in polynomial time \cite{Gabow1983}. The same holds
for $b_t(D)$. Since only polynomially many pairs of
degree bounds need to be considered, $\mathcal L$ is
computable in polynomial time.
\end{proposition}
\begin{proof}
See Appendix~\ref{proof of prop:bmatching-cores}.
\end{proof}

\begin{definition}[Degree second moment]
\label{def:degree-second-moment}
For any explored edge set $F$, define its degree second moment as
\[
P(F)=\sum_{v\in V}\deg_F(v)^2,
\]
where $\deg_F(v)$ denotes the total degree of $v$ in $(V,F)$,
counting the directed edges $uv$ and $vu$ separately when
both belong to $F$.
\end{definition}
The following elementary inequality improves the maximum-degree estimate whenever degrees are sufficiently diffuse.

\begin{lemma}[Degree-second-moment vertex-cover bound]\label{lem:vc-second-moment}
For every nonempty finite undirected graph $H$,
\[
\operatorname{vc}(H)\ge
\frac{|E(H)|^2}{\sum_{z\in V(H)}\deg_H(z)^2}.
\]
\end{lemma}
\begin{proof}
For any vertex cover $C$, every edge has an endpoint in $C$, whence
\[
|E(H)|\le\sum_{z\in C}\deg_H(z)
\le\sqrt{|C|\sum_{z\in V(H)}\deg_H(z)^2}
\]
by Cauchy--Schwarz. Minimize $|C|$ and rearrange.
\end{proof}

\begin{lemma}[Tensor second-moment bound]\label{lem:tensor-second-moment}
In the setting of Theorem~\ref{thm:tensor_cross_edges}, for low-query edge sets $G_s,G_t$ with $M_s=|E(G_s)|$ and $M_t=|E(G_t)|$, the number $q$ of present cross edges satisfies
\[
q\ge\frac{M_s^2M_t^2}{P(E(G_s))P(E(G_t))}.
\]
\end{lemma}
\begin{proof}
For the auxiliary tensor graph $H$, $|E(H)|=M_sM_t$, and $\deg_H(x,y)\le\deg_{G_s}(x)\deg_{G_t}(y)$. Thus
\[
\sum_{(x,y)\in V(H)}\deg_H(x,y)^2
\le P(E(G_s))P(E(G_t)).
\]
The present cross edges form a vertex cover of $H$, so Lemma~\ref{lem:vc-second-moment} yields the claim.
\end{proof}

\begin{corollary}[Second-moment bounds for explored cores]
\label{cor:second-moment-cores}
Let $F_s\subseteq E_s^{(1)}$ and $F_t\subseteq E_t^{(1)}$ be
nonempty first-exploration subsets, and write
\[
m_s=|F_s|,\qquad m_t=|F_t|,
\qquad P_s=P(F_s),\qquad P_t=P(F_t).
\]
Then every correct randomized algorithm $A$ with expected
query complexity $Q$ satisfies
\[
\boxed{
Q\ge
\min\left\{
\frac{m_s}{16},
\frac{m_t}{16},
\frac{m_s^2m_t^2}{64P_sP_t}
\right\}.
}
\]
\end{corollary}

\begin{proof}
Apply the high-query edge deletion argument of
Theorem~\ref{thm:explored-cores} to $F_s$ and $F_t$, then apply
Lemma~\ref{lem:tensor-second-moment} to the remaining low-query subsets,
using the fact that deleting edges cannot increase
their degree second moments.
\end{proof}

\begin{theorem}[Instance-dependent second-moment bound]\label{thm:bid-second-moment}
Let $P_s=P(E_s^{(1)})$ and $P_t=P(E_t^{(1)})$. Every good algorithm $A$ satisfies
\[
Q=\Omega\!\left(\min\left\{S,\frac{S^4}{P_sP_t}\right\}\right).
\]
Consequently,
\[
\boxed{\quad\rho_{\mathrm{BID}}(G,s,t)
=O\!\left(\max\left\{1,\frac{P_sP_t}{S^3}\right\}\right).\quad}
\]
In particular, $P_sP_t=O(S^3)$ is sufficient for instance optimality.
\end{theorem}
\begin{proof}
If $Q\ge S/128$, there is nothing to prove. Otherwise, retaining first explorations and deleting the records whose underlying edge has query probability greater than $1/4$ deletes at most $4Q<S/32$ first-exploration records per side. Since $|E_s^{(1)}|,|E_t^{(1)}|\ge\lfloor S/2\rfloor/2$, the retained low-query sets have $M_s,M_t\ge S/16$ for $S\ge8$; bounded smaller cases affect only the implicit constant. Their degree second moments are no greater than $P_s,P_t$, since deletion only decreases degrees. Lemma~\ref{lem:tensor-second-moment} and Lemma~\ref{lem: cross edge queries} imply
\[
Q\ge\frac q4\ge\frac{M_s^2M_t^2}{4P_sP_t}
\ge\frac{S^4}{262144P_sP_t}.
\]
The claimed asymptotic bounds follow by combining the cases and using $T_{\BID}=O(S)$.
\end{proof}

Recall that the coefficient of variation of a sequence of
positive values is the ratio of its standard deviation to
its mean, measuring the relative dispersion of the values.
In particular, a uniformly bounded coefficient of variation
means that this relative dispersion does not grow with the
size of the graph.

\begin{corollary}[Bounded degree variation]\label{cor:bounded-degree-variation}
Suppose the coefficients of variation of the \emph{nonzero} degree sequences of $(V,E_s^{(1)})$ and $(V,E_t^{(1)})$ are uniformly $O(1)$ throughout a class of instances. Then \BID{} is instance optimal on that class.
\end{corollary}

In other words, BID is instance optimal whenever the degrees of the vertices involved in each exploration do not become increasingly dispersed relative to their average. This permits heterogeneous degrees but excludes explored degree distributions whose relative variation grows without bound.

\begin{proof}
See Appendix~\ref{proof of cor:bounded-degree-variation}.
\end{proof}

Interestingly, the undirected lower-bound construction of Bertram et al. (Figure 4) \cite{BertramJensenThorupWangYan2026} illustrates the sharpness of our second-moment bound. In the regime of a growing instance-optimality gap, the coefficient of variation remains bounded in the forward explored graph, while in the backward explored graph it grows as \(\Theta(\Delta^{2/3}/n^{1/3})\). Moreover, the construction satisfies
\[
\frac{P_sP_t}{S^3}
=\Theta\left(\frac{\Delta^{2/3}}{n^{1/3}}\right),
\]
matching the asymptotic instance-optimality gap established by Bertram et al. Thus, our second-moment bound is asymptotically tight on this construction.

\begin{definition}[First-exploration orientations]
\label{def:first-exploration-orientations}
Orient each first-exploration edge query from the vertex issuing
the query toward the revealed neighbor, as in the proof
of Theorem~\ref{thm:tensor_cross_edges}. Only the first exploration of each underlying edge is included in this construction; reverse-incidence queries still count toward the total query complexity.

For the resulting directed explored graphs, let
$d_s^+(v)$ and $d_s^-(v)$ denote the out- and in-degrees
of $v$ in the forward exploration, respectively, and
define $d_t^+(v)$ and $d_t^-(v)$ analogously for the
backward exploration.

These orientations describe only the execution of the
algorithm; the underlying input graph remains undirected.
\end{definition}

\begin{lemma}[Exact auxiliary degree]\label{lem:tensor-exact-degree}
For the tensor graph $H$ associated with two vertex-disjoint low-query edge sets $G_s,G_t$, the auxiliary vertex $(x,y)$ has degree
\[
\deg_H(x,y)=d_s^+(x)d_t^+(y)+d_s^-(x)d_t^-(y).
\]
\end{lemma}
\begin{proof}
An auxiliary edge incident to $(x,y)$ arises either by pairing two explored edges with tails $x,y$, or by pairing two explored edges with heads $x,y$. These possibilities are disjoint, and distinct pairs of incidence records are counted separately, including when they yield parallel auxiliary edges. Counting them establishes the identity.
\end{proof}

\begin{corollary}[Orientation-sensitive maximum degree]\label{cor:tensor-alignment}
In the setup of Theorem~\ref{thm:tensor_cross_edges}, define
\[
D_{\rm align}=
\max_{x\in X,\,y\in Y}
\bigl(d_s^+(x)d_t^+(y)+d_s^-(x)d_t^-(y)\bigr).
\]
Then $q\ge M_sM_t/D_{\rm align}$, with the convention that the bound is invoked only for nonempty low-query edge sets. In particular, $D_{\rm align}\le\Delta_s\Delta_t$.
\end{corollary}
\begin{proof}
Theorem~\ref{thm:tensor_cross_edges} gives $q\ge\operatorname{vc}(H)\ge |E(H)|/\Delta(H)$. Lemma~\ref{lem:tensor-exact-degree} identifies $\Delta(H)=D_{\rm align}$, and $|E(H)|=M_sM_t$.
\end{proof}

A second-moment refinement also retains the orientation information. On the \emph{first-exploration} edge sets $E_s^{(1)},E_t^{(1)}$, let
\[
\begin{aligned}
A_s&=\sum_x(d_s^+(x))^2,&B_s&=\sum_x(d_s^-(x))^2,&
C_s&=\sum_xd_s^+(x)d_s^-(x),\\
A_t&=\sum_y(d_t^+(y))^2,&B_t&=\sum_y(d_t^-(y))^2,&
C_t&=\sum_yd_t^+(y)d_t^-(y),
\end{aligned}
\]
and set
\begin{equation}\label{eq:psi-tensor}
\Psi=A_sA_t+B_sB_t+2C_sC_t.
\end{equation}
By Lemma~\ref{lem:tensor-exact-degree}, $\Psi$ is precisely
the sum of squared degrees of the tensor graph constructed
from the complete first-exploration edge sets, counting all
incidence-record pairs with multiplicity, even if the sets
have common vertices before low-query edges are deleted.

\begin{corollary}[Orientation-sensitive moment bound]
\label{thm:orientation-moment}
Every good algorithm $A$ satisfies
\[
\boxed{
Q=\Omega\!\left(\min\left\{S,\frac{S^4}{\Psi}\right\}\right),
\qquad
\rho_{\mathrm{BID}}(G,s,t)
=O\!\left(\max\left\{1,\frac{\Psi}{S^3}\right\}\right).
}
\]
In particular, $\Psi=O(S^3)$ is sufficient for instance
optimality.
\end{corollary}

\begin{proof}
Apply the high-query edge deletion argument of
Theorem~\ref{thm:bid-second-moment}. The resulting
low-query tensor graph has $\Omega(S^2)$ edges and
sum of squared degrees at most $\Psi$, so
Lemma~\ref{lem:vc-second-moment}, together with
Lemma~\ref{lem: cross edge queries}, gives
$Q=\Omega(S^4/\Psi)$ whenever $Q<S/128$.
The remaining case is immediate.
\end{proof}

For an indication of the possible gain, consider abstract explored stars with $r$ edges each, oriented outward from the center on the forward side and inward to the center on the backward side. The auxiliary tensor graph is $K_{r,r}$ (plus isolated vertices), so its vertex cover has size $r$, whereas the cruder product-of-maximum-degrees bound gives only $1$. This illustrates the tensor counting mechanism; it is not asserted that every abstract pair of oriented graphs is realizable as the exploration of a particular shortest-path instance.

\begin{definition}[Auxiliary degree profile]
\label{def:auxiliary-degree-profile}
Let $H_0$ be the auxiliary tensor graph constructed from the
first-exploration edge sets $E_s^{(1)},E_t^{(1)}$, preserving their exploration
orientations. Let
\[
h_1\ge h_2\ge\cdots\ge h_N
\]
denote the degrees of $H_0$ in nonincreasing order,
including isolated vertices and counting parallel edges
with multiplicity.

For $S\ge8$, define
\[
J=\min\left\{j:\sum_{i=1}^{j}h_i\ge S^2/256\right\}.
\]
For the finitely many cases $S<8$, set $J=0$.
This index is well-defined since
\[
\sum_{i=1}^{N}h_i
=2|E(H_0)|
=2|E_s^{(1)}||E_t^{(1)}|
=\Omega(S^2)\qquad(S\ge8).
\]
\end{definition}

\begin{theorem}[Degree-profile lower bound]\label{thm:tensor-degree-profile}
Every good algorithm $A$ satisfies
\[
\boxed{\quad Q\ge\min\left\{\frac S{128},\frac J4\right\}.\quad}
\]
In particular, $J=\Omega(S)$ is a sufficient condition for instance optimality.
\end{theorem}
\begin{proof}
See Appendix~\ref{proof of thm:tensor-degree-profile}.
\end{proof}

\paragraph{Comparison with the orientation-sensitive moment bound.}
For $S\ge8$, the sum of squared degrees of $H_0$ is
$\Psi=\sum_{i=1}^{N}h_i^2$ by
Lemma~\ref{lem:tensor-exact-degree} and
Equation~\eqref{eq:psi-tensor}. By the definition of $J$
and Cauchy--Schwarz,
\[
\frac{S^2}{256}
\le \sum_{i=1}^{J}h_i
\le \sqrt{J\sum_{i=1}^{J}h_i^2}
\le \sqrt{J\Psi}.
\]
Thus
\[
J\ge\frac{S^4}{65536\,\Psi}.
\]
In particular, $\Psi=O(S^3)$ implies
$J=\Omega(S)$: the degree-profile sufficient condition
subsumes the orientation-sensitive second-moment
condition. More quantitatively,
Theorem~\ref{thm:tensor-degree-profile} also gives
$Q\ge\min\{S/128,S^4/(262144\,\Psi)\}$.
The moment certificate remains useful because $\Psi$
can be computed from six sums over explored vertex
degrees, without constructing or sorting the entire
auxiliary degree profile.

\begin{proposition}[Degree-profile separation from explored cores and moments]
\label{prop:degree-profile-separation}
There is a family of connected undirected, strictly
positive-weight simple graphs, with fixed adjacency
orders and corrected \BID{} exploration counts $S$,
such that
\[
J=\Omega(S),\qquad
\mathcal L(E_s,E_t)=o(S),\qquad
\frac{\Psi}{S^3}\longrightarrow\infty.
\]
Thus the degree-profile certificate establishes
instance optimality on this family, while the
optimized explored-core and orientation-sensitive
second-moment sufficient conditions do not.
The family also has $S/n\to\infty$ and
$|\mathcal C_{s,t}|=1$, so neither the reservoir
nor the weight-critical-edge condition certifies
instance optimality.
\end{proposition}
\begin{proof}
See Appendix~\ref{app:degree-profile-separation}.
\end{proof}

The tensor degrees of $H_0$ can be computed without
explicitly materializing its edges, using the formula
in Lemma~\ref{lem:tensor-exact-degree}. These
degree-based certificates do not, however, exploit
all the structure of the auxiliary graph. One could
apply matching or vertex-cover methods directly to
$H$ for stronger bounds, but the reduction itself
requires only the existence of a sufficiently large
vertex cover.

\subsection{Necessary structural conditions for large
instance-optimality gaps}
\label{sec:necessary-structure}

The preceding results provide sufficient conditions for
instance optimality in terms of the structure of the explored
edge sets. We now derive a complementary necessary condition
for an algorithm to substantially outperform \BID.

The key observation is that the tensor construction also
provides a lower bound in terms of the matching numbers
of the low-query subgraphs.

\begin{lemma}[Matching lower bound for the tensor graph]
\label{lem:tensor-matching-product}
Consider the setup of
Theorem~\ref{thm:tensor_cross_edges}, and let $H$ be
the corresponding auxiliary tensor graph.

Write $\nu_s$ and $\nu_t$ for the maximum matching
cardinalities of $G_s$ and $G_t$, respectively.
Then
\[
\operatorname{vc}(H)\ge \nu_s\nu_t,
\]
where $\operatorname{vc}(H)$ denotes the minimum
vertex-cover cardinality of $H$.
Consequently, every correct randomized algorithm $A$
satisfies
\[
\nu_s\nu_t\le 4T_A(G,s,t).
\]
\end{lemma}

\begin{proof}
See Appendix~\ref{proof of lem:tensor-matching-product}.
\end{proof}

\begin{theorem}[Necessary vertex-cover concentration]
\label{thm:necessary-cover-concentration}
Let $(G,s,t)$ be an undirected simple positive-weight
instance with $S\ge4$. Let $A$ be any randomized
algorithm that is correct with probability at least
$0.9$ on every input.

Define the low-query edge sets
\[
F_s=
\left\{
e\in E_s^{(1)}:
\Pr[A\text{ queries }e]\le\frac14
\right\},
\]
and
\[
F_t=
\left\{
e\in E_t^{(1)}:
\Pr[A\text{ queries }e]\le\frac14
\right\}.
\]

Suppose
\[
Q<\frac{S}{128}.
\]

Then
\[
\operatorname{vc}(V,F_s)
\le256(n-1)\frac QS
\]
and
\[
\operatorname{vc}(V,F_t)
\le256(n-1)\frac QS.
\]

In particular, for any family of instances satisfying
\[
Q=o(S),
\]
both low-query spanning graphs admit vertex covers of
size $o(n)$ and independent sets of size $n-o(n)$.
\end{theorem}

\begin{proof}
See Appendix~\ref{proof of thm:necessary-cover-concentration}.
\end{proof}

\begin{corollary}[Structural necessity for a
superconstant instance-optimality gap]
\label{cor:necessary-explored-concentration}

Consider a family of undirected simple positive-weight
instances $(G,s,t)$ for which
\[
\frac{T_{\BID}(G,s,t)}
{\tau(G,s,t)}
\longrightarrow\infty,
\]
with $\tau$ as defined in the preliminaries.
Then there exist sets
\[
C_s,C_t\subseteq V(G)
\]
such that
\[
|C_s|,|C_t|=o(n),
\]
and
\[
\left|
\left\{
uv\in E_s:
u,v\notin C_s
\right\}
\right|=o(S),
\]
\[
\left|
\left\{
uv\in E_t:
u,v\notin C_t
\right\}
\right|=o(S).
\]

Equivalently, in each exploration, all but a
vanishing fraction of the explored edges are
incident to a sublinear number of vertices.
\end{corollary}

\begin{proof}
See Appendix~\ref{proof of cor:necessary-explored-concentration}.
\end{proof}

The corollary provides a necessary structural
condition for a superconstant failure of instance
optimality. In particular, the explored edge
sets must become increasingly concentrated around
small vertex sets, after discarding a vanishing
fraction of their edges.

This complements Corollary~\ref{cor:bounded-explored-cores},
which establishes instance optimality when both
explorations contain sufficiently large edge sets
with sufficiently small maximum degrees.

The result does not imply that the full explored
graphs, or the subgraphs induced by all vertices
accessed by \BID, have sublinear vertex covers.
The discarded $o(S)$ records may, for example,
contain a matching of $\Theta(n)$ disjoint edges
when $n=o(S)$. Covering that matching already
requires $\Theta(n)$ vertices, even if all remaining
explored edges are covered by $o(n)$ vertices.

Moreover, the necessary concentration condition
does not characterize non-instance-optimality:
it is not sufficient for the existence of an
algorithm that substantially outperforms \BID.

\section{Further Structural Regimes and Open Problems}
\label{sec:open-structural-regimes}

Bertram et al.~\cite{BertramJensenThorupWangYan2026}
resolve the original unrestricted instance-optimality question
for positively weighted simple graphs in the principal
adjacency-list models. In particular, in the undirected
order-dependent model with degree queries, their
counterexamples show that corrected bidirectional Dijkstra
is not instance optimal on the entire input class; their
nearly tight worst-case ratio bounds quantify this
failure. This settles the original yes-or-no question,
but does not exhaust the structural or metric conditions
under which bidirectional Dijkstra remains instance optimal.
Their Table~1 summarizes several global-parameter and
easy-instance guarantees, rather than an exhaustive list
of all instance-optimal subclasses.

Table~\ref{tab:additional-regimes} collects further
\emph{sufficient} regimes provided by the present analysis.
Throughout the table, graphs are undirected and simple
with strictly positive weights; the vertex universe is
fixed and known, adjacency lists have a fixed order,
and degree queries are available. The query count $S$
includes both searches' full incidence records,
$P_s=P(E_s^{(1)})$, $P_t=P(E_t^{(1)})$,
$\Psi$ is the orientation-sensitive tensor moment,
$J$ is the auxiliary degree-profile index, and
$\mathcal L$ is the optimal explored-core parameter.
All the displayed bounds hold against every universally
$0.9$-correct randomized competing algorithm.
Each row supplies a sufficient condition for
$\rho_{\mathrm{BID}}=O(1)$; the conditions overlap and
are not claimed to be independent or necessary.

\begin{table}[t]
\centering
\footnotesize
\renewcommand{\arraystretch}{1.4}
\resizebox{\linewidth}{!}{%
\begin{tabular}{@{}l@{\hspace{12pt}}l@{\hspace{12pt}}l@{}}
\hline
\textbf{Regime (reference)} &
\textbf{Every good $A$: lower bound on $Q$} &
\textbf{BID: upper bound on $\rho_{\mathrm{BID}}$} \\
\hline
$S=O(n)$ {\scriptsize (Thm.~\ref{thm:quantitative_reservoir})} &
$\Omega\!\left(S/(1+S/n)\right)$ &
$O(1+S/n)$ \\
$\mathcal L=\Omega(S)$ {\scriptsize (Thm.~\ref{thm:explored-cores}; Prop.~\ref{prop:bmatching-cores})} &
$\ge \mathcal L/16$ &
$O(S/\mathcal L)$ \\
$P_sP_t=O(S^3)$ {\scriptsize (Thm.~\ref{thm:bid-second-moment})} &
$\Omega\!\left(\min\{S,S^4/(P_sP_t)\}\right)$ &
$O\!\left(\max\{1,P_sP_t/S^3\}\right)$ \\
$\Psi=O(S^3)$ {\scriptsize (Cor.~\ref{thm:orientation-moment})} &
$\Omega\!\left(\min\{S,S^4/\Psi\}\right)$ &
$O\!\left(\max\{1,\Psi/S^3\}\right)$ \\
$J=\Omega(S)$ {\scriptsize (Thm.~\ref{thm:tensor-degree-profile})} &
$\ge \min\{S/128,J/4\}$ &
$O\!\left(\max\{1,S/J\}\right)$ \\
$h_D=O(1)$, $D=O(\sqrt n)$ {\scriptsize (Cor.~\ref{cor:few-exceptional-degrees})} &
$\Omega\!\left(S/(1+h_D+D^2/n)\right)$ &
$O(1+h_D+D^2/n)$ \\
$|\mathcal C_{s,t}|=\Omega(S)$ {\scriptsize (Prop.~\ref{prop:critical-weight-edges})} &
$\ge \tfrac45|\mathcal C_{s,t}|$ &
$O(S/|\mathcal C_{s,t}|)$ \\
\hline
\end{tabular}%
}
\caption{Illustrative sufficient conditions for instance optimality of
corrected bidirectional Dijkstra on undirected, positive-weight simple
graphs in the fixed-universe, fixed-adjacency-order model with degree
queries. Notation follows Section~\ref{sec:Preliminaries};
the table assumes $S\ge8$ and universally correct
randomized competitors. The second column gives universal \emph{lower bounds} on $Q$;
the third gives corresponding \emph{upper bounds} on BID's ratio.
The conditions are sufficient but not necessary, and some overlap.
The reservoir-type dependence in the first row was already
implicit in earlier work.}
\label{tab:additional-regimes}
\end{table}

The degree-tail and weight-critical-edge conditions
encode distinct kinds of information: one depends
on the global degree distribution, while the other
depends on weights and the chosen terminals.

These examples motivate a broader question:
\emph{Which structural and metric restrictions restore
instance optimality, and can one obtain tight quantitative
bounds in terms of more informative parameters than
$n$, $m$, and $\Delta$?} An even stronger goal is an
efficiently computable structural characterization,
up to constant factors, of the instance-optimality
ratio of corrected bidirectional Dijkstra on
individual inputs. Finding additional sufficient
conditions alone cannot achieve such a characterization:
matching converse statements or algorithms realizing
the possible speedups would also be required.

Our necessary concentration result gives a partial
restriction on instances permitting superconstant
speedups. In addition, Proposition~\ref{prop:critical-weight-edges}
implies that every family with $\tau(G,s,t)=o(S)$ must
satisfy $|\mathcal C_{s,t}|=o(S)$.
Thus any such family must exhibit both strong
concentration of explored edges, after discarding
$o(S)$ records, and a vanishing proportion of
individually weight-critical edges. These are necessary,
not sufficient, conditions. Understanding which
concentrated configurations admit shorter universal
query certificates, and how topology and edge weights
interact in those configurations, remains open.
Extensions of these instance-sensitive criteria to
directed graph models are another natural direction.
Finally, testing the certificates on routing-related
networks and on deliberately dense or heterogeneous
graph families would help distinguish practical
obstructions from extremal examples.

\section{Conclusion}\label{sec:conclusion}

The instance-optimality of corrected bidirectional Dijkstra
depends on both graph structure and the information
exposed by its query oracle. We established a
constant-factor guarantee for directed and undirected
simple graphs under local discovery. For undirected
simple graphs in the standard fixed-universe model,
we derived universal execution-sensitive lower bounds
from vertex reservoirs, low-query edge switches, and
weight-critical edges. Optimized explored cores,
degree moments, exploration orientations, and the
auxiliary tensor degree profile yield complementary
sufficient conditions. We also gave a necessary
concentration condition for a superconstant gap:
after discarding $o(S)$ records per direction,
the remaining explored edges must be incident
to $o(n)$ vertices on each side.

These conditions do not yet characterize which
concentrated instances admit faster universal
query certificates. The classification remains
open for individual fixed-order executions even
though the unrestricted model-wide questions are
settled by the concurrent work of
Bertram et al.~\cite{BertramJensenThorupWangYan2026}.
The local-discovery guarantee and the
undirected fixed-universe certificates must not
be conflated; extending the latter to directed
settings is a separate problem, as discussed
in Section~\ref{sec:open-structural-regimes}.


\bibliographystyle{plainnat}
\bibliography{references}


\appendix

\section{Additional Proofs for the Multigraph Case}
\label{app:multigraph}
\subsection{Proof of Lemma~\ref{m closed}}
\label{proof of m closed}

\begin{proof}
We first address the terminal cases without changing the argument below. The case $m=s$ is impossible: every valid candidate using the forward edge $(u^s_\mu,s)$ has length at least
$d(s,u^s_\mu)+\ell(u^s_\mu,s)+d(s,t)>d(s,t)=\mu$,
since the edge weights are positive and the backward tentative label $\hat d(s,t)$ cannot be less than the true distance. The case $m=t$ is permitted, and the argument below applies without modification: the backward exploration of $t$ starts with the exact label $d(t,t)=0$ and reveals the edge $(t,u^s_\mu)$. If $u^s_\mu$ is already closed, that query establishes the optimal candidate; otherwise the corrected update rule establishes it no later than the forward relaxation that gives $u^s_\mu$ its final tentative distance (or the backward query if that happens later).

For the sake of contradiction, assume that $m$ is already closed and explored by the backward execution and that the $\mu \gets d(s,u^s_\mu)+\ell (u^s_\mu, m)+\hat d(m,t)$ update had not yet been done. Since $u^s_\mu$ is a neighbor of $m$ the backward execution must already have set a finite value for $\hat d (u^s_\mu, t)$ via $d(m,t) + \ell (u^s_\mu,m)$. Now we will consider two options. 
    
    First, consider that $u^s_\mu$ is already closed when the backward execution finished exploring edges from $m$. Then it must be the case that the backward execution would update $\mu$ as $\mu \gets d(m,t)+\ell (u^s_\mu,m) + d(s,u^s_\mu)$ which contradicts our assumption that the final update happened from the forward execution as $\mu \gets d(s,u^s_\mu)+\ell (u^s_\mu, m)+\hat d(m,t)$.

    The remaining option is that $u^s_\mu$ is not yet closed when the backward execution finishes exploring $m$. In that case it follows that the backward execution will set $\hat d(u^s_\mu,t) \gets d(m,t) + \ell (u^s_\mu,m)$ before $u^s_\mu$ is closed. If the forward execution has already computed the final value for $\hat d(s,u^s_\mu)$ at that time, then when the backward execution explores the edge $(u^s_\mu,m)$ it will set the final value for $\mu$ as $\mu \gets d(m,t)+\ell (u^s_\mu,m)+\hat d(s,u^s_\mu)$ again contradicting the assumption that the update happened from the forward execution. Finally, if the forward execution has not yet computed the final value for $\hat d(s,u^s_\mu)$ then consider the node that will set the final value for $\hat d (s,u^s_\mu)$. Let this node be $s'$. In that case, when the update $\hat d(s,u^s_\mu) \gets d(s,s')+\ell (s',u^s_\mu)$ occurs, the update $\mu \gets d(s,s') + \ell(s',u^s_\mu) + \ell(u^s_\mu,m) + d(m,t)$ would also occur right after, since with the relaxation on Line~\ref{line: stop condition} we do not require $u^s_\mu$ to be closed in the backward execution. This contradicts the assumption that $\mu$ was updated as $\mu \gets d(s,u^s_\mu) + \ell (u^s_\mu,m) + d(m,t)$.

    Since all alternatives lead to contradictions the assumption that $m$ is already closed and explored from the backward execution and the final $\mu$ update had not yet been done must be false. 
\end{proof}

\subsection{Proof of Theorem~\ref{omission_optimality}}
\label{proof of omission_optimality}
\begin{proof} 
    \textbf{Correctness.} By the structure of \BID it is clear that for every edge $uv$ the algorithm will consider $\hat d (s,u)+\ell(u,v)+\hat d(v,t)$ as a shortest path candidate at the latest when the last of $u$ and $v$ is closed and explored by their respective executions, given that the algorithm does not terminate beforehand. Secondly, since nodes are closed in non-decreasing distance order and the edge weights are positive, it is clear that when nodes $u_s$ and $u_t$ are closed, all subsequent $st$-path candidates will have their lengths of at least $d(s,u_s)+d(u_t,t)$. It follows that when $u_s$ and $u_t$ are closed, all path candidates of length less than $d(s,u_s)+d(u_t,t)$ have already been examined. Finally, it follows that when the condition on Line~\ref{line: stoping conditon} is satisfied the shortest path had already been considered, and $\mu$ was updated accordingly on Line~\ref{line: mu update}. This concludes the correctness proof.
    
    \textbf{Instance Optimality.} By alternation,
    $\bigl||E_s|-|E_t|\bigr|\le1$.
    
    Suppose, for the sake of contradiction, that a correct algorithm $A$
    makes at most $S/32$ queries in expectation. For a physical edge $e$, let $Q_e$
    indicate the event that $A$ reveals $e$ from \emph{either} of its
    incidence entries. An underlying edge is represented by at most two
    records in each explored collection, so
    \[
      \sum_{e\in E_s}\Pr[Q_e=1]\le2T_A(G,s,t),
      \qquad
      \sum_{e\in E_t}\Pr[Q_e=1]\le2T_A(G,s,t).
    \]
    Since the two searches alternate, for $S\ge4$ each collection
    contains at least $S/3$ records. Averaging gives an explored
    forward record $u_1v_1\in E_s$ and backward record
    $u_2v_2\in E_t$ for which
    \[
      \Pr[Q_{u_1v_1}=1],\ \Pr[Q_{u_2v_2}=1]
      \le\frac{2(S/32)}{S/3} =\frac{3}{16}<\frac14.
    \]
    Let $a_1$ be the vertex from which the forward record was queried,
    and let $b_2$ be the vertex from which the backward record was queried.
    Lemma~\ref{explored edge distance} gives
    $d(s,a_1)+d(b_2,t)<\mu$. In the directed case, label the selected
    arcs as $u_1v_1$ and $u_2v_2$, where $u_1=a_1$ and $v_2=b_2$.
    In the undirected case, instead label the endpoints so that
    $d(s,u_1)\le d(s,v_1)$ and $d(v_2,t)\le d(u_2,t)$.
    Thus, in either case,
    $d(s,u_1)+d(v_2,t)\le d(s,a_1)+d(b_2,t)<\mu$.
    Relabeling does not change the probability of revealing either
    underlying edge. The bounded cases $S<4$ require only a constant
    number of queries by \BID{} and can be handled directly.
    Write $e_1=u_1v_1$ and $e_2=u_2v_2$ for the two selected underlying edges.
    In the case that the selected records refer to the same underlying edge, we construct $G'$ by simply taking $G$ and decreasing the weight of $u_1v_1$ to some $\delta < \mu - d(s,u_1) - d(v_2,t)$.
    In the case that they refer to distinct underlying edges, we construct $G'$ by taking the same graph $G$, but replacing the edges $u_1v_1$ and $u_2v_2$ with $u_1v_2$ and $u_2v_1$. In the case of undirected graphs this clearly does not affect the degree of the nodes. In the case of directed graphs orient the edges appropriately to preserve in and out degrees, namely replace $u_1v_1, u_2v_2$ with $u_1v_2, u_2v_1$. Since we have $d(s,u_1) + d(v_2,t) < \mu$ we set the weight of $u_1v_2$ to be some $\delta < \mu - d(s,u_1) - d(v_2,t)$, and the edge weight of $u_2v_1$ to be arbitrary. We first verify that the removal of $e_1$ and $e_2$
does not destroy the shortest paths from $s$ to $u_1$
and from $v_2$ to $t$.

By the endpoint choices above, no shortest $s u_1$-path
contains $e_1$: for undirected edges this follows from the
minimum-distance choice of $u_1$, and for directed arcs from
positivity and the orientation $u_1\to v_1$.

Moreover, if such a path contained $e_2$, then
$d(s,v_2)\le d(s,u_1)$, implying
\[
d(s,t)
\le d(s,v_2)+d(v_2,t)
\le d(s,u_1)+d(v_2,t)
<d(s,t),
\]
a contradiction.

By symmetry, no shortest $v_2 t$-path contains
either $e_1$ or $e_2$.

Consequently, both shortest-path segments remain
present in the modified graph $G'$. The new connector $u_1v_2$ therefore yields an $st$-path shorter than the original optimum.
    By the union bound,
    \[
    \Pr[Q_{u_1v_1}=1 \text{ or } Q_{u_2v_2}=1]
    \le
    \frac{1}{2},
    \]
    and therefore
    \[
    \Pr[Q_{u_1v_1}=0 \text{ and } Q_{u_2v_2}=0]
    \ge
    \frac{1}{2}.
    \]
    The only queries whose answers differ between $G$ and $G'$ are queries involving $u_1v_1$ or $u_2v_2$ on $G$, or equivalently, queries involving $u_1v_2$ or $u_2v_1$ on $G'$. Therefore, with probability at least $1/2$, the executions of $A$ on $G$ and $G'$ are identical, assuming the same internal randomness. On this event, $A$ returns the same answer on both inputs even though the correct shortest $st$-paths differ.

    Consequently, $A$ is incorrect on at least one of $G$ and $G'$ with probability at least \(0.5\cdot 0.5=0.25\).
    This contradicts the requirement that a correct algorithm succeeds with probability at least $0.9$, equivalently that it errs with probability at most $0.1$.
    
    Therefore every correct algorithm must perform $\Omega(|E_s|+|E_t|)$ queries in expectation. Since \BID performs $O(|E_s|+|E_t|)$ queries, it is instance optimal.
\end{proof}

\section{Additional Proofs for Simple Graphs and General Certificates}
\label{app:simple}

\subsection{Proof of Proposition~\ref{prop:critical-weight-edges}}
\label{proof of prop:critical-weight-edges}
\begin{proof}
Fix $e\in\mathcal C_{s,t}$, choosing the orientation
$u\to v$ in the undirected case so that
$d_s(u)+d_t(v)<\mu$. Change only its weight
to $\varepsilon>0$ satisfying
\[
d_s(u)+\varepsilon+d_t(v)<\mu.
\]
The original shortest $su$- and $vt$-paths
and the modified edge form an $st$-walk of
length less than $\mu$ in the modified graph:
lowering an edge weight cannot lengthen
either original path. Thus the correct
shortest-path distances differ.

Couple $A$ on the two inputs using identical
random coins. All query answers are the same
unless an incidence entry of $e$ is read.
Therefore the transcripts, and hence the
outputs, agree with probability at least
$1-p_A(e)$. As $A$ has error probability
at most $0.1$ on either input,
$1-p_A(e)\le0.2$, proving
$p_A(e)\ge4/5$. Each adjacency query reveals
at most one underlying edge or arc, so
\[
Q\ge\sum_{e\in\mathcal C_{s,t}}p_A(e)
\ge\frac45|\mathcal C_{s,t}|.
\]
In models where both BID search directions
can be performed, $T_{\BID}(G,s,t)\le2S+2$,
giving the stated instance-optimality
consequence.
\end{proof}

\subsection{Proof of Theorem~\ref{thm:bid-local-discovery}}
\label{proof of thm:bid-local-discovery}
\begin{proof}
Fix an arbitrary labeling by opaque vertex handles and
an arbitrary ordering of every adjacency list. Let $A$ be a
universally correct randomized competitor. For this
proof abbreviate $a=|E_s|$, $b=|E_t|$. By alternation,
$|a-b|\le1$.
On a directed graph the forward search queries outgoing
lists and the backward search queries incoming lists.
On either type of graph each search scans at most one
relevant list per closed vertex; each vertex other than
the initial root must have been discovered by a previous
edge query. Thus degree queries contribute at most $S+2$
further queries, and
$T_{\BID}(G,s,t)\le2S+2$.
Assume first that $a,b\ge1$.

Retain one explored record per underlying edge or arc on
each side. An undirected edge occurs at most twice in a
search's records, once per endpoint, whereas a directed arc
occurs at most once in either search. Thus there are at
least $a/2$ distinct selected edges or arcs forward and
$b/2$ backward. Let $p_A(e)$ be the probability that $A$
reveals $e$ through any incidence entry, including
incoming or outgoing records in the directed case.
Since a query reveals at most one underlying edge or arc,
the sum of $p_A(e)$ over either distinct-edge collection
is at most $Q$. Averaging yields a forward selected
edge or arc $e_s$ and a backward selected edge or arc
$e_t$ with
\[
p_A(e_s)\le\frac{2Q}{a},
\qquad
p_A(e_t)\le\frac{2Q}{b}.
\]
In the directed case, write $e_s=u_1\to v_1$
and $e_t=u_2\to v_2$: these are scanned from
$u_1$ in the forward search and from $v_2$
in the backward search, respectively.
In the undirected case write $e_s=u_1v_1$
and $e_t=u_2v_2$ with $u_1$ and $v_2$
their respective first-exploring endpoints.
By Lemma~\ref{explored edge distance}, in both cases
\[
d_G(s,u_1)+d_G(v_2,t)<\mu:=d_G(s,t).
\]
If $\mu=+\infty$, the construction below simply
creates a finite $st$-path.

\emph{Four distinct endpoints.}
Suppose $u_1,v_1,u_2,v_2$ are pairwise distinct.
Introduce a single fresh vertex $z$, whose handle
is unavailable to $A$, and delete $e_s,e_t$.
For an undirected graph, join $z$ to all four
endpoints:
$u_1z,\quad v_1z,\quad u_2z,\quad v_2z$.
For a directed graph, instead insert
\[
u_1\to z,\quad u_2\to z,\quad
z\to v_1,\quad z\to v_2.
\]
The resulting graph remains simple. Every original
vertex loses one incident edge or arc and gains
one in its place, preserving its degree; in the
directed case, each original indegree and outdegree
is preserved separately. The new vertex has degree
$4$ in the undirected case, and indegree and
outdegree $2$ in the directed case. Replace each
deleted incidence entry in its original indexed
position by the corresponding new entry, leaving
all other entries at original vertices unchanged.
Give the two edges or arcs on the
$u_1$--$z$--$v_2$ connector positive weights
whose sum is less than
$\mu-d_G(s,u_1)-d_G(v_2,t)$ for finite $\mu$.
The other edges or arcs incident to $z$ receive
arbitrary positive weights.

A shortest $s u_1$-path avoids $e_s$: in the
directed case $e_s$ leaves $u_1$, while in the
undirected case first exploration from $u_1$
implies $d_G(s,u_1)\le d_G(s,v_1)$.
It also avoids $e_t$, since any $s u_1$-path
using this edge or arc passes through $v_2$
and would imply
\[
\mu\le d_G(s,v_2)+d_G(v_2,t)
\le d_G(s,u_1)+d_G(v_2,t)<\mu.
\]
Symmetrically, a shortest $v_2 t$-path avoids
both deleted edges or arcs. Thus these terminal
segments survive, and the new connector produces
an $st$-path strictly shorter than $\mu$
(or a finite path if $\mu=+\infty$).

\emph{Overlapping endpoints.}
The equality $u_1=v_2$ is impossible, since
$d(s,u_1)+d(u_1,t)\ge\mu$.
For directed arcs, every other possible
overlap supplies a directed path from $u_1$
to $v_2$ using one or both selected arcs:
$u_1=u_2$ gives $e_t$, $v_1=v_2$ gives
$e_s$, and $v_1=u_2$ gives $e_s$ followed
by $e_t$. The same observations work for
undirected edges; they also cover identical
selected edges or arcs. Lower the positive
weights of the relevant one or two selected
edges or arcs so that the resulting
$u_1$--$v_2$ connector is sufficiently short.
No degree or incidence position changes.

In every case, only queries revealing $e_s$
or $e_t$ can distinguish the original graph
from the modified one. In particular, fresh
handles cannot be accessed directly.
Couple the two executions of $A$ with
identical internal randomness. Then
\[
\Pr[\text{transcripts differ}]
\le p_A(e_s)+p_A(e_t)
\le2Q\left(\frac1a+\frac1b\right).
\]
The correct shortest-path distances on the
two graphs differ, so a common transcript
forces an erroneous answer on at least
one graph. Because $A$ errs with probability
at most $0.1$ on either, the probability of
identical transcripts is at most $0.2$.
Consequently,
\[
2Q\left(\frac1a+\frac1b\right)\ge\frac45,
\qquad
Q\ge\frac{2/5}{1/a+1/b}.
\]
Combining this with $T_{\BID}\le2S+2$,
we obtain
\[
\frac{T_{\BID}(G,s,t)}{Q}
\le\frac52(2S+2)\left(\frac1a+\frac1b\right)
\le30.
\]
Indeed, with $k=\min\{a,b\}\ge1$, the
counts are either $(k,k)$ or $(k,k+1)$.
In either case,
\[
(2S+2)\left(\frac1a+\frac1b\right)
\le8+\frac4k\le12.
\]

If $a=0$ or $b=0$, alternation implies
$S\le1$ and $T_{\BID}\le4$.
A query-free competitor cannot distinguish
two inputs with the same initial terminal
handles but different shortest-path answers.
Universal $0.9$-correctness therefore
requires at least one query with probability
at least $0.8$, and $Q\ge0.8$.
The claimed ratio also follows in this case.

No distributional assumption on the fixed
adjacency-list ordering was used. The same
pointwise guarantee holds for every ordering
and labeling by opaque handles. Averaging it
over randomized representations establishes
the order-oblivious bound as well.
\end{proof}

\subsection{Proof of Theorem~\ref{thm:reservoir_node}}
\label{proof of thm:reservoir_node}
\begin{proof}

Let $\mu=d(s,t)$. We prove the contrapositive. Suppose that
\(
T_A(G,s,t)\le S/32.
\)
If $S<4$, universal $0.9$-correctness on instances with
$s\ne t$ implies $T_A(G,s,t)\ge4/5$: without making a
query, the algorithm cannot distinguish two admissible
inputs with the same initial information but different
shortest-path answers. Thus $T_A(G,s,t)>S/32$ in this
case, a contradiction. We may therefore assume $S\ge4$.
Since the relaxed bidirectional Dijkstra execution alternates edge explorations, we have
\(
|E_s|=|E_t|\pm 1.
\)
Thus $|E_s|\ge S/3$ and $|E_t|\ge S/3$.
Since each underlying edge occurs in at most two oriented incidence records of $E_s$, the average probability that $A$ reveals an underlying edge represented in $E_s$ is at most
\[
\frac{2T_A(G,s,t)}{|E_s|}
\le
\frac{2(S/32)}{S/3}
=
\frac{3}{16}
<
0.25.
\]
Therefore, there exist underlying edges $e_s,e_t$ represented
by forward and backward explored records, respectively, with
\[
\Pr[A\text{ queries }e_s]\le 1/4,\qquad
\Pr[A\text{ queries }e_t]\le 1/4.
\]
Let $a_1$ and $b_2$ denote the vertices from which the selected
records were queried. By Lemma~\ref{explored edge distance},
$d(s,a_1)+d(b_2,t)<\mu$. Since $G$ is undirected, label the
endpoints of $e_s=u_1v_1$ and $e_t=u_2v_2$ so that
$d(s,u_1)\le d(s,v_1)$ and $d(v_2,t)\le d(u_2,t)$.
Consequently,
$d(s,u_1)+d(v_2,t)\le d(s,a_1)+d(b_2,t)<\mu$.
The relabeling does not change either low-query probability.
Let
$\Delta=\mu-d(s,u_1)-d(v_2,t)>0$.

We first dispose of the case in which the two chosen edges are not vertex-disjoint. If $u_1=v_2$, then
$d(s,u_1)+d(u_1,t)<\mu$,
which is impossible. Hence any overlap gives a path from $u_1$ to $v_2$ using one or both of the edges $u_1v_1$ and $u_2v_2$. Indeed, if $v_1=v_2$, then the edge $u_1v_1$ connects $u_1$ to $v_2$; if $u_1=u_2$, then the edge $u_2v_2$ connects $u_1$ to $v_2$; and if $v_1=u_2$, then the two-edge path
$u_1-v_1-u_2-v_2$
connects $u_1$ to $v_2$ after identifying $v_1=u_2$.

In this case, construct $G'$ by lowering the weights of the one or two involved edges so that the resulting path from $u_1$ to $v_2$ has total length less than $\Delta$. Then $G'$ contains an $st$-path of length strictly smaller than $\mu$. The modification can be discovered only if $A$ queries $u_1v_1$ or $u_2v_2$. Therefore, by the union bound, the probability that $A$ discovers the modification is at most
\(
0.25+0.25=0.5.
\)
With probability at least $0.5$, $A$ has not queried any edge revealing the new shorter path. Coupling the executions of $A$ on $G$ and $G'$ with the same random choices, on this event the transcript is identical and $A$ returns the same output on both graphs. Since the shortest $st$-path length is different in $G$ and $G'$, $A$ is incorrect on at least one of the two inputs with probability at least $0.25>0.1$, contradicting correctness. Thus we may assume from now on that the four vertices $u_1,v_1,u_2,v_2$ are distinct.

Because $u_1$ and $v_2$ were chosen as the respective
closer endpoints, removing $u_1v_1$ or $u_2v_2$ cannot destroy
shortest paths from $s$ to $u_1$ or from $v_2$ to $t$;
the cross-edge argument below also excludes either segment's
using the other removed edge. Where only edge weights are decreased,
the same original path segments remain available.

We now construct a shortcut through the low-access vertex $r$. The construction has two independent parts: one connecting $u_1$ to $r$, and one connecting $r$ to $v_2$.

First we describe the forward connector from $u_1$ to $r$ using the edge $u_1v_1$.

If $r=u_1$, no forward modification is needed, and the connector from $u_1$ to $r$ has length $0$.

If $r=v_1$, we use the edge $u_1v_1$ as the connector and lower its weight.

Assume now that $r\notin{u_1,v_1}$. If $u_1r\in E(G)$, we use the existing edge $u_1r$ and lower its weight. If $u_1r\notin E(G)$ but $v_1r\in E(G)$, we use the path
$u_1-v_1-r$
and lower the weights of $u_1v_1$ and $v_1r$. Finally, if neither $u_1r$ nor $v_1r$ is an edge of $G$, we remove the edge $u_1v_1$ and add the two edges $u_1r$ and $rv_1$. The edge $u_1r$ is used in the new shortcut and is assigned a small positive weight, while $rv_1$ is assigned an arbitrary positive weight. This preserves the degrees of $u_1$ and $v_1$; only the degree and incidence list of $r$ may change.

In every forward case, we obtain a path from $u_1$ to $r$ whose total weight can be made arbitrarily small. Moreover, every modified oracle answer outside the incidence list of $r$ is contained in the original edge $u_1v_1$.

We define the backward connector from $r$ to $v_2$ symmetrically using the edge $u_2v_2$, which was explored from $v_2$ by the backward execution.

If $r=v_2$, no backward modification is needed, and the connector from $r$ to $v_2$ has length $0$.

If $r=u_2$, we use the edge $u_2v_2$ as the connector and lower its weight.

Assume now that $r\notin{u_2,v_2}$. If $rv_2\in E(G)$, we use the existing edge $rv_2$ and lower its weight. If $rv_2\notin E(G)$ but $ru_2\in E(G)$, we use the path
$r-u_2-v_2$
and lower the weights of $ru_2$ and $u_2v_2$. Finally, if neither $rv_2$ nor $ru_2$ is an edge of $G$, we remove the edge $u_2v_2$ and add the two edges $rv_2$ and $ru_2$. The edge $rv_2$ is used in the new shortcut and is assigned a small positive weight, while $ru_2$ is assigned an arbitrary positive weight. This preserves the degrees of $u_2$ and $v_2$; only the degree and incidence list of $r$ may change.

Again, in every backward case, we obtain a path from $r$ to $v_2$ whose total weight can be made arbitrarily small. Moreover, every modified oracle answer outside the incidence list of $r$ is contained in the original edge $u_2v_2$.

Combining the forward and backward connectors, we obtain in $G'$ a path of the form
$s\leadsto u_1 \leadsto r \leadsto v_2 \leadsto t$.
Choose the modified positive edge weights so that the total length of the connector from $u_1$ to $v_2$ through $r$ is less than $\Delta$. Then
\[
d(s,u_1)
+
d_{G'}(u_1,r)
+
d_{G'}(r,v_2)
+
d(v_2,t)
<
d(s,u_1)+\Delta+d(v_2,t)
=
\mu.
\]
Thus $G'$ has an $st$-path strictly shorter than the shortest $st$-path in $G$.

We now bound the probability that $A$ discovers the modification. By construction, every changed oracle answer is revealed only if $A$ queries $u_1v_1$, queries $u_2v_2$, or accesses the vertex $r$. Therefore
\[
\begin{aligned}
\Pr[A\text{ discovers the modification}]
&\le
\Pr[A\text{ queries }u_1v_1]
+
\Pr[A\text{ queries }u_2v_2]
+
\Pr[A\text{ accesses }r] \\
&\le
0.25+0.25+0.25
=
0.75.
\end{aligned}
\]
Consequently, with probability at least $0.25$, $A$ does not query any oracle entry that reveals the new shorter path. Coupling the executions of $A$ on $G$ and $G'$ using the same random choices, on this event the transcript seen by $A$ is the same on the two graphs. Therefore $A$ returns the same output on $G$ and $G'$, while the correct shortest $st$-path length is different. Hence $A$ is incorrect on at least one of $G$ and $G'$ with probability at least
$\frac{0.25}{2} = 0.125 > 0.1$.
This contradicts the assumption that $A$ is correct with probability at least $0.9$ on every input.

Therefore the assumption
\(
T_A(G,s,t)\le S/32
\)
is impossible. Hence
$T_A(G,s,t)>\frac{|E_s|+|E_t|}{32}$,
as claimed.
\end{proof}

\subsection{Proof of Theorem~\ref{thm:quantitative_reservoir}}
\label{proof of thm:quantitative_reservoir}
\begin{proof}
Fix a correct randomized algorithm $A$.

An undirected edge can occur at most twice in either explored
edge collection, once from each endpoint. Since each neighbor
query reveals at most one edge, averaging gives an edge
$e_s\in E_s$ and an edge $e_t\in E_t$ such that
\[
\Pr[A\text{ queries }e_s]\le\frac{2Q}{a},
\qquad
\Pr[A\text{ queries }e_t]\le\frac{2Q}{b}.
\]

A query accesses at most two vertices, where a vertex is
considered accessed if its degree or an incident edge is queried.
Therefore, some vertex $r\in V(G)$ satisfies
$\Pr[A\text{ accesses }r]\le\frac{2Q}{n}$.

Apply the indistinguishability construction from the proof
of Theorem~\ref{thm:reservoir_node} to $e_s,e_t,r$.

That construction produces an instance $G'$ with a strictly
shorter $st$-path, preserving all oracle answers except those
revealed by querying $e_s$, querying $e_t$, or accessing $r$.
Overlapping endpoints are handled as in the proof of
Theorem~\ref{thm:reservoir_node}.

Couple the executions of $A$ on $G$ and $G'$ using identical
random choices. By the union bound, the probability that $A$
discovers a modified oracle answer is at most
\[
p\le
2Q\left(\frac1a+\frac1b+\frac1n\right).
\]

If $p<4/5$, then with probability greater than $1/5$ the
two executions have identical transcripts.

Since the correct shortest-path distances differ, the sum
of the error probabilities of $A$ on $G$ and $G'$ would exceed
$1/5$. This contradicts the requirement that the error
probability on each input be at most $1/10$.

Consequently,
$2Q\left(\frac1a+\frac1b+\frac1n\right) \ge\frac45$,
which gives
$Q\ge\frac{2/5}{1/a+1/b+1/n}$.

We now derive the instance-optimality ratio.

Since \BID alternates edge explorations, we have
$|a-b|\le1$.
Moreover, \BID performs $S$ neighbor queries and at most
$S+2$ degree queries. Therefore,
$T_{\BID}(G,s,t)\le2S+2$.

For integers $a,b\ge1$ satisfying $|a-b|\le1$,
$(2S+2)\left(\frac1a+\frac1b\right)\le12$.

Taking the infimum over all correct randomized algorithms
therefore gives
\[
\begin{aligned}
\frac{T_{\BID}(G,s,t)}{\tau(G,s,t)}
&\le
\frac52(2S+2)
\left(\frac1a+\frac1b+\frac1n\right)\\
&\le
30+\frac{5(S+1)}n.
\end{aligned}
\]

Finally, if $a=0$ or $b=0$, alternation implies $S\le1$,
so \BID performs at most four queries.

On nontrivial shortest-path inputs with unrestricted positive
weights, any algorithm that is $0.9$-correct on every input
must make at least one query with probability at least $0.8$:
a query-free response cannot distinguish two
inputs with the same terminal handles but different answers.

Hence $\tau(G,s,t)\ge0.8$, and the displayed ratio bound
also holds in these cases.
\end{proof}

\subsection{Proof of Lemma~\ref{lem:at_most_one_boundary_pair}}
\label{proof of lem:at_most_one_boundary_pair}
\begin{proof}
Suppose, for the sake of contradiction, that the algorithm closes two pairs of vertices \((p_s^1,p_t^1)\) and \((p_s^2,p_t^2)\)
where $p_s^1,p_s^2$ are closed by the forward execution, $p_t^1,p_t^2$ are closed by the backward execution, $p_s^1\neq p_s^2$, and $p_t^1\neq p_t^2$, such that
\(
d(s,p_s^1)+d(p_t^1,t)\ge \mu
\)
and
\(
d(s,p_s^2)+d(p_t^2,t)\ge \mu .
\)
Without loss of generality, assume that the final update of $\mu$ occurs in the forward execution through an edge $(u_\mu^s,m)$, so that
$\mu = d(s,u_\mu^s)+\ell(u_\mu^s,m)+\hat d(m,t)$.

Let $(p_s,p_t) \in \{(p_s^1,p_t^1),(p_s^2,p_t^2)\}$ be the pair in which both of the nodes close first.
Since the two original pairs have distinct forward endpoints and distinct backward endpoints, at least one node of the pairs remains to be closed after both $p_s$ and $p_t$ have been closed. Now we will show that the algorithm terminates when the later of $p_s$ and $p_t$ is closed.

By assumption, the pair $(p_s,p_t)$ satisfies
\(
    d(s,p_s)+d(p_t,t)\ge \mu.
\)
Substituting the final value of $\mu$, we obtain
\[
    d(s,p_s)+d(p_t,t)
    \ge
    d(s,u_\mu^s)+\ell(u_\mu^s,m)+\hat d(m,t).
\]

We now distinguish two cases.

First, suppose that
\(
    d(s,p_s)>d(s,u_\mu^s).
\)
Since vertices are closed in nondecreasing distance order in the forward execution, $u_\mu^s$ is closed and explored before $p_s$ is closed. Hence the final update of $\mu$ has already occurred by the time $p_s$ is closed. Therefore, when the later of $p_s$ and $p_t$ is closed, the final value of $\mu$ has already been set. Since $p_s$ and $p_t$ are closed in the forward and backward executions respectively, the current vertices $u_s$ and $u_t$ satisfy
\[
    \hat d(s,u_s)+\hat d(u_t,t)
    \ge
    d(s,p_s)+d(p_t,t)
    \ge
    \mu.
\]
Thus the algorithm terminates by the stopping condition on Line~\ref{line: termination}.

Second, suppose that
\(
    d(s,p_s)\le d(s,u_\mu^s).
\)
Then the inequality above implies
\[
    d(p_t,t)
    \ge
    d(s,u_\mu^s)-d(s,p_s)
    +
    \ell(u_\mu^s,m)+\hat d(m,t)
    \ge
    \ell(u_\mu^s,m)+\hat d(m,t).
\]
Since edge weights are positive, this gives
\(
    d(p_t,t)>\hat d(m,t).
\)
Hence, in the backward execution, the vertex $m$ is closed before $p_t$. By Lemma~\ref{m closed}, once $m$ has been closed, the final update of $\mu$ must already have occurred. Therefore, by the time $p_t$ is closed, the final value of $\mu$ has already been set. As before, when the later of $p_s$ and $p_t$ is closed, we have
\[
    \hat d(s,u_s)+\hat d(u_t,t)
    \ge
    d(s,p_s)+d(p_t,t)
    \ge
    \mu,
\]
and so the algorithm terminates by the stopping condition on Line~\ref{line: termination}.

In both cases, the algorithm terminates when the later of $p_s$ and $p_t$ is closed. This prevents it from closing the remaining endpoint of the second pair, contradicting the assumption that two pairs with distinct forward and backward endpoints were both closed. Therefore, any two pairs satisfying the inequality must share their forward endpoint or their backward endpoint.

\end{proof}

\subsection{Proof of Lemma~\ref{lem: required cross edges}}
\label{proof of lem: required cross edges}
\begin{proof}
Suppose neither $xy$ nor $x'y'$ exists.
By Lemma~\ref{lem:low_query_intersection},
the selected edges $xx'$ and $yy'$ have four distinct
endpoints. Their first-exploration orientations give
$d(s,x)\le d(s,x')$ and $d(y,t)\le d(y',t)$,
while Lemma~\ref{explored edge distance} gives
$d(s,x)+d(y,t)<d(s,t)$. Thus the undirected
edge-switching construction in the proof of
Theorem~\ref{omission_optimality} applies:
replace $xx',yy'$ by $xy,x'y'$ and give $xy$
sufficiently small positive weight. The original
shortest $sx$- and $yt$-paths survive the switch.
The new graph remains simple because the four
endpoints are distinct and both replacement edges
were absent. All degrees are preserved, and only
entries revealing the two selected low-query edges
change. Each is queried by $A$ with probability
at most $1/4$, so the coupling from that proof
yields identical transcripts with probability
at least $1/2$ despite different shortest-path
distances. One of the two inputs therefore has
error probability at least $1/4>0.1$,
contradicting universal correctness.
\end{proof}

\subsection{Proof of Lemma~\ref{lem:low_query_intersection}}
\label{proof of lem:low_query_intersection}
\begin{proof}
Recall that $G_s$ and $G_t$ contain only edges queried by $A$ with probability at most $1/4$.

We first show that
\(
E(G_s)\cap E(G_t)=\emptyset.
\)
Suppose, for contradiction, that there exists an edge $e=uv$ such that
\(
e\in E(G_s)\cap E(G_t).
\)
Let $c_s$ be the endpoint from which $e$ was explored by the forward execution, and let $c_t$ be the endpoint from which $e$ was explored by the backward execution.

We first note that $c_s\neq c_t$. Indeed, if $c_s=c_t$, then the same vertex is closed in both executions. Since
\(
d(s,c_s)+d(c_s,t)\ge \mu,
\)
Lemma~\ref{lem:at_most_one_boundary_pair} implies that when this vertex is closed by the second execution, BID terminates before exploring any edge from it. Hence $e$ could not be explored from $c_s$ in both executions, a contradiction.

Thus $c_s\neq c_t$. Since $e$ was explored from $c_s$ in the forward execution and from $c_t$ in the backward execution, Lemma~\ref{explored edge distance} gives
\(
d(s,c_s)+d(c_t,t)<\mu.
\)
Construct a graph $G'$ from $G$ by lowering the weight of $e$ to a positive value $\delta$ satisfying
\(
\delta<\mu-d(s,c_s)-d(c_t,t).
\)
Then $G'$ contains an $st$-path of length strictly smaller than $\mu$, namely the path obtained by concatenating a shortest $s c_s$-path, the edge $e$, and a shortest $c_t t$-path.

The only oracle answers changed between $G$ and $G'$ are those revealing the weight of $e$. Since $e\in E(G_s)\cap E(G_t)$, the edge $e$ is queried by $A$ with probability at most $1/4$. Therefore, with probability at least $3/4$, algorithm $A$ does not query any changed oracle entry. Coupling the executions of $A$ on $G$ and $G'$ using the same internal randomness, $A$ has the same transcript and returns the same output on both inputs on this event. Since the shortest $st$-path length differs between $G$ and $G'$, $A$ is incorrect on at least one of the two inputs with probability at least
\(
0.75\cdot 0.5
=
0.375
>
0.1.
\)
This contradicts the assumption that $A$ is correct with probability at least $0.9$ on every input. Hence
\(
E(G_s)\cap E(G_t)=\emptyset.
\)

It remains to show that
\(
V(G_s)\cap V(G_t)=\emptyset.
\)
Suppose, for contradiction, that there exists a vertex
\(
r\in V(G_s)\cap V(G_t).
\)
Since $r\in V(G_s)$, there exists an edge $e_s\in E(G_s)$ incident to $r$. Since $r\in V(G_t)$, there exists an edge $e_t\in E(G_t)$ incident to $r$. By the first part of the proof, $e_s\neq e_t$.

Let $c_s$ be the endpoint from which $e_s$ was explored by the forward execution, and let $c_t$ be the endpoint from which $e_t$ was explored by the backward execution. We claim that
\(
c_s\neq c_t.
\)
If $c_s=c_t=r$, then $r$ is closed in both executions. As above, Lemma~\ref{lem:at_most_one_boundary_pair} implies that BID terminates when $r$ is closed by the second execution, before exploring any edge from it. This contradicts the fact that both $e_s$ and $e_t$ are explored from $r$ in the two executions.

If $c_s=c_t\neq r$, then, because $G$ is simple and both $e_s$ and $e_t$ are incident to $r$, we must have
\(
e_s=e_t,
\)
contradicting $E(G_s)\cap E(G_t)=\emptyset$. Therefore $c_s\neq c_t$.

Since $e_s$ was explored from $c_s$ in the forward execution and $e_t$ was explored from $c_t$ in the backward execution, Lemma~\ref{explored edge distance} gives
\(
d(s,c_s)+d(c_t,t)<\mu.
\)
Let
\(
\Delta=\mu-d(s,c_s)-d(c_t,t)>0.
\)

We now modify the weights of the edges connecting $c_s$ to $r$ and $r$ to $c_t$. If $c_s=r$, the first connector is trivial and has length $0$; otherwise the first connector is the edge $e_s$. If $c_t=r$, the second connector is trivial and has length $0$; otherwise the second connector is the edge $e_t$.

Construct $G'$ by lowering the weights of the nontrivial connector edges among $e_s$ and $e_t$ so that their total length is less than $\Delta$. Then $G'$ contains an $st$-path of length strictly smaller than $\mu$:
\[
d(s,c_s)
+
d_{G'}(c_s,r)
+
d_{G'}(r,c_t)
+
d(c_t,t)
<
\mu.
\]

The only oracle answers changed between $G$ and $G'$ are those revealing the modified edges, which are contained in ${e_s,e_t}$. Both $e_s$ and $e_t$ belong to the low-query subgraphs, so
\[
\Pr[A\text{ queries }e_s]\le \frac{1}{4},
\qquad
\Pr[A\text{ queries }e_t]\le \frac{1}{4}.
\]
By the union bound,
\[
\Pr[A\text{ queries }e_s\text{ or }e_t]\le \frac{1}{2}.
\]
Hence, with probability at least $1/2$, algorithm $A$ does not query any changed oracle entry. Coupling the executions of $A$ on $G$ and $G'$ using the same internal randomness, $A$ has the same transcript and returns the same output on both inputs on this event. Since the shortest $st$-path length differs between $G$ and $G'$, $A$ is incorrect on at least one of the two inputs with probability at least
\(
0.5\cdot0.5
=
0.25
>
0.1.
\)
This contradicts the assumption that $A$ is correct with probability at least $0.9$ on every input. Therefore
\(
V(G_s)\cap V(G_t)=\emptyset.
\)

Finally, if an edge belonged to both $E(G_s)$ and $E(G_t)$, then its endpoints would belong to both $V(G_s)$ and $V(G_t)$, contradicting the vertex-disjointness just proved. Thus
\(
E(G_s)\cap E(G_t)=\emptyset.
\)
\end{proof}

\subsection{Proof of Corollary~\ref{cor:degree_bounded_by_bid_search}}
\label{proof of cor:degree_bounded_by_bid_search}
\begin{proof}
Suppose $Q<S/1024$. For $S\ge8$, each first-exploration
set has at least $\lfloor S/2\rfloor/2\ge S/8$ edges.
Deleting records whose underlying edge has query probability
greater than $1/4$ removes at most $4Q$ records per side.
Thus $M_s,M_t\ge S/16$. Since no underlying edge
occurs more than once in either first-exploration set,
their maximum degrees are at most $\Delta$.
The tensor bound gives
\[
q\ge\frac{M_sM_t}{\Delta^2}
\ge\frac{S^2}{256\Delta^2}\ge\frac S{256}.
\]
By Lemma~\ref{lem: cross edge queries},
$Q\ge q/4\ge S/1024$, a contradiction.
When $S<8$, the standard nonzero-query lower bound
$Q\ge0.8$ proves the claim directly.
\end{proof}

\subsection{Proof of Theorem~\ref{thm:explored-cores}}
\label{proof of thm:explored-cores}
\begin{proof}
If $Q\ge\min\{m_s,m_t\}/16$, the claim follows.

Otherwise, suppose
$Q<\frac{\min\{m_s,m_t\}}{16}$.

From each of $F_s$ and $F_t$, remove every explored edge
whose underlying undirected edge is queried by $A$ with
probability greater than $1/4$.

Recall that each undirected edge can be represented by
at most two directed edges in either explored set.
Consequently, for any set $F$ of explored edges,
\[
\sum_{e\in F}p_A(e)\le 2Q,
\]
where $p_A(e)$ denotes the probability that $A$ queries
either orientation of the underlying edge.

Since every removed edge contributes more than $1/4$
to this sum, at most $8Q$ edges are removed from
either set.

Denote the remaining sets by $F_s'$ and $F_t'$.
By our assumption on $Q$,
\[
|F_s'|\ge m_s-8Q>\frac{m_s}{2},
\qquad
|F_t'|\ge m_t-8Q>\frac{m_t}{2}.
\]

Since $F_s'$ and $F_t'$ consist entirely of low-query
explored edges, the tensor argument of
Theorem~\ref{thm:tensor_cross_edges} applies to these
subgraphs. Writing $q$ for the number of cross edges
between their vertex sets, we obtain
\[
q\ge
\frac{|F_s'||F_t'|}
{\Delta(F_s')\Delta(F_t')}
\ge
\frac{m_sm_t}{4D_sD_t},
\]
where we used
$\Delta(F_s')\le D_s$ and $\Delta(F_t')\le D_t$.

By Lemma~\ref{lem: cross edge queries},
$Q\ge\frac q4 \ge\frac{m_sm_t}{16D_sD_t}$.
Combining this with the case
$Q\ge\min\{m_s,m_t\}/16$
proves the theorem.
\end{proof}

\subsection{Proof of Corollary~\ref{cor:few-exceptional-degrees}}
\label{proof of cor:few-exceptional-degrees}
\begin{proof}
Write $h=h_D$.
If $S<32(hn+D^2+1)$, the quantitative reservoir
bound (Theorem~\ref{thm:quantitative_reservoir}) gives
\[
\rho_{\mathrm{BID}}
\le30+\frac{5(S+1)}n
=O\left(1+h+\frac{D^2}{n}\right).
\]
Otherwise remove from each first-exploration set
$E_i^{(1)}$ every edge incident to one of the $h$
exceptional vertices, and call the remaining sets
$F_s,F_t$. At most $h(n-1)$ underlying edges are
removed from either search. Since the searches alternate
incidence queries and each underlying edge appears
at most twice in a single search,
\[
|F_s|,|F_t|
\ge\frac{S-1}{4}-h(n-1)
\ge\frac S8.
\]
Furthermore, $\Delta(F_s),\Delta(F_t)\le D$,
and $S\ge32D^2$. Theorem~\ref{thm:explored-cores}
therefore yields, for every good competing algorithm,
\[
Q\ge\frac1{16}
\min\left\{\frac S8,\frac{S^2}{64D^2}\right\}
\ge\frac S{128}.
\]
As $T_{\BID}\le2S+2$, the ratio is bounded by an
absolute constant in this case as well.
\end{proof}

\subsection{Proof of Proposition~\ref{prop:bmatching-cores}}
\label{proof of prop:bmatching-cores}
\begin{proof}
For any chosen $F_s,F_t$, use $d_s=\Delta(F_s)$ and $d_t=\Delta(F_t)$; the maximum-cardinality solutions under these caps have at least as many edges and can only increase the displayed minimum. Conversely, for every degree-cap pair, choose subgraphs attaining the two $b$-matching cardinalities. Their actual maximum degrees are at most the caps, so their value in the definition of $\mathcal L$ is at least the displayed minimum. This proves equality. The edge-selection problem defining $b_F(d)$ is a standard cardinality $b$-matching problem with vertex capacities $d$, allowing the two incidence records of an underlying edge to appear as parallel edges.
\end{proof}

\subsection{Proof of Corollary~\ref{cor:bounded-degree-variation}}
\label{proof of cor:bounded-degree-variation}
\begin{proof}
For either explored edge graph $(V,F)$, let $m=|F|$
denote its number of edges, $r$ its number of nonisolated
vertices, and $P=P(F)$ its degree second moment.

Let $\mu=2m/r$ be the mean nonzero degree, using the
handshaking lemma. Since the coefficient of variation
is $c=\sigma/\mu$, where $\sigma$ is the standard
deviation, we have
$\frac{P}{r} =\sigma^2+\mu^2 =\mu^2(1+c^2)$.
Consequently,
$P=\frac{4m^2}{r}(1+c^2)$.
Both explored edge graphs have a simple underlying
undirected graph and contain $\Theta(S)$ edge records;
hence both have $r=\Omega(\sqrt S)$ nonisolated vertices.
Their moments therefore satisfy $P_sP_t=O(S^3)$.
Apply Theorem~\ref{thm:bid-second-moment}.
\end{proof}

\subsection{Proof of Theorem~\ref{thm:tensor-degree-profile}}
\label{proof of thm:tensor-degree-profile}
\begin{proof}
If $Q\ge S/128$, we are done. Otherwise, delete from $E_s^{(1)},E_t^{(1)}$ the records whose underlying edge is queried with probability above $1/4$. This removes at most $4Q<S/32$ records per side. For $S\ge8$, the first-exploration sets contain at least $\lfloor S/2\rfloor/2\ge S/8$ records each, leaving $M_s,M_t\ge S/16$. Thus the remaining tensor graph $H'$ is a subgraph of $H_0$ and $|E(H')|=M_sM_t\ge S^2/256$. The present cross edges form a vertex cover $C$ of $H'$. Therefore,
\[
\frac{S^2}{256}\le |E(H')|
 \le\sum_{z\in C}\deg_{H'}(z)
 \le\sum_{i=1}^{|C|}h_i.
\]
Thus $|C|\ge J$. Lemma~\ref{lem: cross edge queries} now yields $Q\ge |C|/4\ge J/4$.
\end{proof}

\subsection{Proof of Lemma~\ref{lem:tensor-matching-product}}
\label{proof of lem:tensor-matching-product}
\begin{proof}
Let $M_s$ and $M_t$ be maximum matchings of $G_s$
and $G_t$, respectively.

Recall that every pair of edges
$e_s\in M_s,\qquad e_t\in M_t$
produces an edge of the auxiliary tensor graph $H$.

Since the edges within each matching are vertex-disjoint,
the auxiliary edges corresponding to distinct pairs
$(e_s,e_t)$ are also vertex-disjoint.

Thus, $H$ contains a matching of cardinality
$|M_s||M_t|=\nu_s\nu_t$.

Every vertex cover must contain at least one endpoint
of each edge in a matching. Hence
\[
\operatorname{vc}(H)\ge\nu_s\nu_t.
\]

By Theorem~\ref{thm:tensor_cross_edges}, the present
cross edges between $V(G_s)$ and $V(G_t)$ form
a vertex cover of $H$.

Writing $q$ for their number, we therefore have
\[
q\ge\operatorname{vc}(H)\ge\nu_s\nu_t.
\]

Finally, Lemma~\ref{lem: cross edge queries}
establishes that every such cross edge is queried
by $A$ with probability at least $1/4$.
Consequently,
$T_A(G,s,t)\ge\frac q4\ge\frac{\nu_s\nu_t}{4}$.
Rearranging proves the claim.
\end{proof}

\subsection{Proof of Theorem~\ref{thm:necessary-cover-concentration}}
\label{proof of thm:necessary-cover-concentration}
\begin{proof}
Let $F_s,F_t$ be the low-query first-exploration
sets in the statement. At most $4Q$ records are
discarded from either first-exploration set.
When $S\ge8$, each first-exploration set has at
least $\lfloor S/2\rfloor/2\ge S/8$ records.
Consequently, since $Q<S/128$,
\[
|F_s|,|F_t|\ge S/8-4Q\ge S/16.
\]
Let $\nu_s,\nu_t$ denote their maximum matching
cardinalities. By Lemma~\ref{lem:tensor-matching-product},
$\nu_s\nu_t\le4Q$. Since an underlying simple edge
is retained at most once in each first-exploration set,
their maximum degrees are at most $n-1$.
The endpoints of a maximum matching cover all edges, so
\[
|F_t|\le2(n-1)\nu_t,\qquad
\nu_t\ge\frac{S}{32(n-1)}.
\]
Thus
\[
\nu_s\le128(n-1)\frac QS,\qquad
\operatorname{vc}(V,F_s)\le2\nu_s
\le256(n-1)\frac QS.
\]
The symmetric inequality follows by exchanging $s$ and $t$.
The complement of a vertex cover is an independent set.
For $4\le S<8$ the claimed numerical inequalities
follow from the nonzero-query bound $Q\ge0.8$
and the fact that each low-query set has at most three edges.
\end{proof}

\subsection{Proof of Corollary~\ref{cor:necessary-explored-concentration}}
\label{proof of cor:necessary-explored-concentration}
\begin{proof}
Since \BID makes $O(S)$ queries, the assumption
\[
\frac{T_{\BID}(G,s,t)}
{\tau(G,s,t)}
\longrightarrow\infty
\]
implies
$\tau(G,s,t)=o(S)$.

For each instance, choose a good algorithm $A$
whose expected query complexity satisfies
$Q=T_A(G,s,t)\le2\tau(G,s,t)$.

Then $Q=o(S)$, so eventually $Q<S/128$.

By Theorem~\ref{thm:necessary-cover-concentration},
the low-query spanning graphs $(V,F_s)$ and $(V,F_t)$
have vertex covers $C_s,C_t$ satisfying
\[
|C_s|,|C_t|
\le256(n-1)\frac QS
=o(n).
\]

Every underlying edge retained in $F_s$ has an endpoint in $C_s$, and so do any later reverse queries of that edge. Consequently, a record of $E_s$ with neither endpoint in $C_s$ must represent an underlying edge omitted from $F_s$. These omitted edges have query probability above $1/4$, and each can give rise to at most two incidence records.

Therefore,
\[
\left|
\left\{
uv\in E_s:u,v\notin C_s
\right\}
\right|
\le
8Q
=o(S).
\]

The backward-exploration statement follows
symmetrically.
\end{proof}

\section{Additional Constructions}
\label{app:constructions}

\subsection{Proof of Proposition~\ref{prop:degree-profile-separation}}
\label{app:degree-profile-separation}

\begin{figure}[tbp]
\centering
\begin{tikzpicture}[
  scale=.85, transform shape,
  dot/.style={circle,fill=black,inner sep=1.7pt},
  group/.style={draw=gray!65,fill=gray!8,rounded corners=2pt,
                align=center,inner sep=3pt,minimum width=12mm},
  link/.style={draw=black!75,line width=.65pt},
  complete/.style={draw=black!75,double,double distance=1.5pt,line width=.4pt},
  lab/.style={font=\scriptsize}
]
\node[dot,label={[lab]left:$s$}] (s) at (-5,0) {};
\node[dot,label={[lab]right:$t$}] (t) at (5,0) {};
\node[group,lab] (U) at (-3.4,1) {$U$\\$h$ vertices};
\node[group,lab] (W) at (-1.6,1) {$W$\\$d$ vertices};
\node[group,lab] (X) at (3.4,1) {$X$\\$d$ vertices};
\node[group,lab] (Y) at (1.6,1) {$Y$\\$h$ vertices};
\node[dot,label={[lab]below:$a$}] (a) at (-3.4,-1.1) {};
\node[dot,label={[lab]below:$b$}] (b) at (3.4,-1.1) {};
\node[group,lab] (A) at (-1.6,-1.1) {$A$\\$q$ leaves};
\node[group,lab] (B) at (1.6,-1.1) {$B$\\$q$ leaves};
\draw[link,line width=1pt] (s)--node[lab,below] {$1$} (t);
\draw[link] (s)--node[lab,above,sloped] {$1/10$} (U);
\draw[link] (t)--node[lab,above,sloped] {$1/10$} (X);
\draw[link] (s)--node[lab,below,sloped] {$1/5$} (a);
\draw[link] (t)--node[lab,below,sloped] {$1/5$} (b);
\draw[complete] (U)--node[lab,above] {$2$} (W);
\draw[complete] (X)--node[lab,above] {$2$} (Y);
\draw[link] (a)--node[lab,above] {$2$} (A);
\draw[link] (b)--node[lab,above] {$2$} (B);
\end{tikzpicture}
\caption{Degree-profile separation. The double links represent
complete bipartite graphs; links to the leaf blocks represent
all $q$ pendant edges. All edges and weights are as shown,
with no additional edges.}
\label{fig:degree-profile-separation}
\end{figure}

\begin{proof}
Fix $k\ge2$ and write $h=k^2$, $d=k^3$, $q=k^4$.
Take disjoint vertex sets $|U|=|Y|=h$, $|W|=|X|=d$,
$|A|=|B|=q$ and vertices $s,t,a,b$, with precisely
the edges and weights shown in
Figure~\ref{fig:degree-profile-separation}. Choose
an adjacency order that scans $st$ first at both terminals.
The graph is simple and connected; $st$ is its only edge
between the two attached sides, so $\mu=d(s,t)=1$ and
$n=2h+2d+2q+4=\Theta(k^4)$.

\emph{Exploration.}
Forward distances to $U,a$ are $1/10,1/5$,
respectively; backward distances to $X,b$ are the
same. All other nonterminal vertices have distance
greater than $1/5$ in the corresponding search.
Until these vertices have been scanned, the BID
stopping sum is at most $2/5<\mu$.
Thus the forward search scans all $hd$ edges of
$U\times W$, and the backward search all $hd$
edges of $X\times Y$.
Before scanning $a$ and $b$, their respective
query counts differ by $O(d)$, while
$q/d=k\to\infty$. Edge alternation therefore
forces each hub to scan $\Theta(q)$ pendant
edges before either search can advance to a
distance-$1$ vertex and terminate. Hence
\[
S=\Theta(hd+q)=\Theta(k^5),\qquad
|E_s^{(1)}|,|E_t^{(1)}|=\Theta(S).
\]
The $U\times W$ edges are first explored
forward from $U$; the $X\times Y$ edges
are first explored backward from $X$.

\emph{Optimized cores.}
Every forward first-exploration edge touches
$U\cup\{s,a\}$, and every backward one touches
$Y\cup\{t,b\}$. Therefore, for all nonempty
$F_s\subseteq E_s^{(1)}$ and $F_t\subseteq E_t^{(1)}$,
\[
|F_s|\le(h+2)\Delta(F_s),\quad
|F_t|\le(h+2)\Delta(F_t),\qquad
\mathcal L(E_s,E_t)\le(h+2)^2=O(k^4)=o(S).
\]

\emph{Second moments.}
Let $q_s,q_t=\Theta(q)$ be the numbers of
pendant edges first queried from $a,b$.
Then $d_s^+(a)=q_s$, $d_t^+(b)=q_t$, and
Lemma~\ref{lem:tensor-exact-degree} gives
$\deg_{H_0}(a,b)\ge q_sq_t$.
Thus
\[
\Psi\ge(q_sq_t)^2=\Omega(q^4)=\Omega(k^{16}),
\qquad S^3=\Theta(k^{15}),
\]
so $\Psi/S^3\to\infty$.

\emph{Degree profile.}
Call $(v,w)\in V(H_0)$ exceptional if
$v\in\{s,a\}$ or $w\in\{t,b\}$.
Writing $M_s=|E_s^{(1)}|$ and $M_t=|E_t^{(1)}|$,
the sum of all exceptional tensor degrees is at most
\[
\begin{aligned}
&\bigl(\deg_{E_s^{(1)}}(s)+\deg_{E_s^{(1)}}(a)\bigr)M_t\\
&\qquad+
 \bigl(\deg_{E_t^{(1)}}(t)+\deg_{E_t^{(1)}}(b)\bigr)M_s
=O(qS)=o(S^2).
\end{aligned}
\]
Away from these coordinate classes, the maximum
forward out-/in-degrees are $d,h$, respectively,
and the maximum backward out-/in-degrees are $h,d$.
Lemma~\ref{lem:tensor-exact-degree} therefore bounds
every nonexceptional tensor degree by
$dh+hd=2hd=O(S)$.
Any $j$ tensor vertices consequently carry at most
$o(S^2)+O(jS)$ degree. Since the largest $J$
degrees sum to at least $S^2/256$, necessarily
\[
J=\Omega(S).
\]
Finally, $S/n=\Theta(k)\to\infty$.
Only the bridge $st$ is individually weight-critical,
so $\mathcal C_{s,t}=\{st\}$.
\end{proof}

\subsection{Proof of Proposition~\ref{prop:clique-hub-separation}}
\label{app:exceptional-hubs-proof}

\begin{figure}[tbp]
\centering
\begin{tikzpicture}[
  scale=.92,transform shape,
  clique/.style={draw=gray!65,fill=gray!8,line width=.75pt},
  dot/.style={circle,fill=black,inner sep=1.7pt},
  leaf/.style={draw=gray!65,fill=white,rounded corners=2pt,inner sep=3pt},
  internal/.style={draw=gray!65,line width=.6pt},
  lab/.style={font=\scriptsize}
]
\draw[clique] (-2.6,0) ellipse [x radius=1.55,y radius=.83];
\draw[clique] (2.6,0) ellipse [x radius=1.55,y radius=.83];
\node[dot,label={[lab]below:$s$}] (s) at (-3.6,-.15) {};
\node[dot,label={[lab]below:$a$}] (a) at (-1.5,-.15) {};
\node[dot,label={[lab]below:$h_s$}] (hs) at (-2.6,.44) {};
\node[dot,label={[lab]above:$t$}] (t) at (3.6,-.15) {};
\node[dot,label={[lab]below:$b$}] (b) at (1.5,-.15) {};
\node[dot,label={[lab]below:$h_t$}] (ht) at (2.6,.44) {};
\draw[internal] (s)--(a) (s)--(hs) (hs)--(a);
\draw[internal] (t)--(b) (t)--(ht) (ht)--(b);
\draw[black,line width=1pt] (a)--node[lab,above] {$10$} (b);
\node[leaf,lab] (Ls) at (-2.6,1.63) {$L$ leaves};
\node[leaf,lab] (Lt) at (2.6,1.63) {$L$ leaves};
\draw[internal] (hs)--node[lab,right] {$100$ each} (Ls);
\draw[internal] (ht)--node[lab,right] {$100$ each} (Lt);
\node[lab] at (-2.6,-1.02) {$C_s\cong K_r$ (weight $1$)};
\node[lab] at (2.6,-1.02) {$C_t\cong K_r$ (weight $1$)};
\end{tikzpicture}
\caption{Clique--hub separation. Each ellipse is a full
$r$-vertex clique (only representative internal edges
are drawn); each hub carries $L$ pendant leaves.}
\label{fig:clique-hub-separation}
\end{figure}

\begin{proof}
Fix $3/2<\alpha<2$ and set
$L=\lfloor r^\alpha\rfloor$.
Take two disjoint cliques $C_s,C_t$ of size $r$,
all internal weights $1$, with $s,a\in C_s$ and
$t,b\in C_t$. Join $a$ to $b$ by an edge of
weight $10$. Attach $L$ leaves of weight $100$
to each of $h_s\in C_s\setminus\{s,a\}$ and
$h_t\in C_t\setminus\{t,b\}$, as in
Figure~\ref{fig:clique-hub-separation}. There are
no other edges. The shortest $st$-distance is $12$;
$n=2r+2L=\Theta(r^\alpha)$ and
$m=r(r-1)+2L+1=\Theta(r^2)$.

Both searches discover the respective clique
vertices at distance $1$. While such vertices
remain unsettled on both sides, the stopping
sum is at most $2<12$. Hence at least one
search scans $\Theta(r)$ clique vertices,
each of degree at least $r-1$. Since
$S\le4m$, it follows that $S=\Theta(r^2)$;
alternation gives $|E_s|,|E_t|=\Theta(r^2)$.
Only $2L+1=o(S)$ edges are outside the cliques.
Retaining the first-explored clique edges
as $F_s,F_t$ therefore gives
\[
|F_s|,|F_t|=\Theta(S),\qquad
\Delta(F_s),\Delta(F_t)=O(r),\qquad
\Delta(F_s)\Delta(F_t)=O(S).
\]
The bounded-degree-core certificate thus establishes
instance optimality.

To see why the untrimmed conditions fail, choose
the fixed tie-breaking orders so that the hubs
$h_s,h_t$ are explored first among distance-$1$
vertices, with their pendant incidences scanned
before their remaining clique incidences.
Both searches then query all $L$ pendant edges.
The maximum degree is
$\Delta=L+r-1=\Theta(r^\alpha)$.
Aside from the hubs, the clique vertices contribute
$O(r^3)$ to each second moment, and the leaves
contribute $O(L)$. As $\alpha>3/2$,
$P_s,P_t=\Theta(L^2+r^3)=\Theta(r^{2\alpha})$.
Consequently,
\[
\frac{\Delta^2}{S}=\Theta(r^{2\alpha-2}),\qquad
\frac{P_sP_t}{S^3}=\Theta(r^{4\alpha-6}),\qquad
\frac{S}{n}=\Theta(r^{2-\alpha}),
\]
all diverging in the stipulated range.

The corresponding untrimmed maximum-degree
and second-moment tensor lower bounds are
$\Theta(r^{4-2\alpha})$ and
$\Theta(r^{8-4\alpha})$, respectively, both
$o(S)$. By contrast, the retained clique
cores give
$|F_s||F_t|/(\Delta(F_s)\Delta(F_t))
=\Omega(S)$. For the illustrative choice
$\alpha=8/5$, the reservoir and untrimmed
second-moment ratio bounds both grow as
$O(r^{2/5})$, while the core criterion yields
$\rho_{\mathrm{BID}}=O(1)$.
\end{proof}

\subsection{Separating the weight-critical-edge certificate from
the degree-based certificates}
\label{app:critical-edge-separation}

The next family separates the weight-sensitive
certificate from the \emph{numerical sufficient
conditions} of the reservoir, optimized cores,
degree moments, orientations, and auxiliary
degree profile, not from the underlying switching
lemma itself.

\begin{proposition}[A weight-sensitive separation family]
\label{prop:critical-separation-family}
There exists a family of connected undirected, strictly
positive-weight simple graphs with terminals $s,t$ and
corrected \BID{} exploration count $S$ such that
\[
|\mathcal C_{s,t}|=\Theta(S),\qquad
S/n\to\infty,\qquad
\mathcal L(E_s,E_t)/S\to0,
\]
\[
P_sP_t/S^3\to\infty,\qquad
\Psi/S^3\to\infty,\qquad J/S\to0.
\]
The maximum-degree and exceptional-degree-tail
conditions also fail to certify instance optimality,
whereas the weight-critical-edge certificate succeeds.
\end{proposition}

\begin{figure}[tbp]
\centering
\begin{tikzpicture}[
  scale=.84,transform shape,
  dot/.style={circle,fill=black,inner sep=1.7pt},
  group/.style={draw=gray!65,fill=gray!8,rounded corners=2pt,
                minimum width=13mm,inner sep=3pt,align=center},
  link/.style={draw=black!75,line width=.65pt},
  complete/.style={draw=black!75,double,double distance=1.5pt,line width=.4pt},
  lab/.style={font=\scriptsize}
]
\node[dot,label={[lab]left:$s$}] (s) at (-5,0) {};
\node[dot,label={[lab]right:$t$}] (t) at (5,0) {};
\node[group,lab] (U) at (-3.45,1.05) {$U$\\$k$ vertices};
\node[group,lab] (W) at (-1.5,1.05) {$W$\\$r$ vertices};
\node[group,lab] (Z) at (1.5,-1.05) {$Z$\\$r$ vertices};
\node[group,lab] (V) at (3.45,-1.05) {$V$\\$k$ vertices};
\node[dot,label={[lab]below:$a$}] (a) at (-3.6,-1.1) {};
\node[dot,label={[lab]above:$b$}] (b) at (3.6,1.1) {};
\draw[link,line width=1pt] (s)--node[lab,above] {$1$} (t);
\draw[link] (s)--node[lab,above,sloped] {$7/20$} (U);
\draw[link] (t)--node[lab,above,sloped] {$7/20$} (V);
\draw[complete] (U)--node[lab,above] {$2/5$} (W);
\draw[complete] (V)--node[lab,below] {$2/5$} (Z);
\draw[link] (W)--node[lab,above,sloped,pos=.6] {$3/5$} (t);
\draw[link] (s)--node[lab,below,sloped,pos=.6] {$3/5$} (Z);
\draw[link] (s)--node[lab,below,sloped] {$9/20$} (a);
\draw[link] (t)--node[lab,above,sloped] {$9/20$} (b);
\end{tikzpicture}
\caption{Weight-critical separation family. Double links
stand for complete bipartite edge sets; single links to
set boxes stand for all the indicated edges.}
\label{fig:weight-critical-separation}
\end{figure}

\begin{proof}
Put $r=k^2$ and take disjoint $|U|=|V|=k$,
$|W|=|Z|=r$, together with $s,t,a,b$.
Include exactly the edges and weights indicated
in Figure~\ref{fig:weight-critical-separation}.
The $U$--$W$ and $V$--$Z$ links are complete
bipartite graphs $K_{k,r}$.
The direct edge $st$ uniquely realizes $\mu=1$;
every competing route through a bipartite region
has length $7/20+2/5+3/5=27/20>1$.

For every $uw\in U\times W$,
$d(s,u)+d(w,t)=19/20<\mu$; the same holds
symmetrically for the $V\times Z$ edges.
The other critical edges are $tw$ ($w\in W$),
$sz$ ($z\in Z$), and $st$. Hence
\[
|\mathcal C_{s,t}|=2kr+2r+1.
\]

In corrected BID, the forward search scans $s$,
then $U$ (distance $7/20$), then $a$
(distance $9/20$); backward it scans
$t,V,b$ at the corresponding distances.
At this stage the last closed distances sum
to $9/10<\mu$. Closing the next vertex at
distance $3/5$ immediately stops the algorithm,
since $3/5+9/20>1$, before that vertex's
adjacency list is scanned. Therefore,
independently of within-level tie-breaking,
\[
S=2kr+2r+4k+6=\Theta(kr),\qquad
n=2r+2k+4=\Theta(r).
\]
It follows that $|\mathcal C_{s,t}|=\Theta(S)$
and $S/n=\Theta(k)\to\infty$.

The global maximum degree is
$\Delta=r+k+2=\Theta(r)$, so
$\Delta^2/S=\Theta(r/k)=\Theta(k)$.
Each of the $k$ vertices in $U$ has at least
$r$ forward first-exploration edges, and each
vertex in $V$ has at least $r$ backward ones.
Thus $P_s,P_t\ge kr^2$, yielding
$P_sP_t/S^3=\Omega(r/k)=\Omega(k)$.
All $U\times W$ and $V\times Z$ records
are oriented outward from $U,V$, respectively;
the $k^2$ tensor vertices in $U\times V$
each have degree at least $r^2$. Consequently,
\[
\Psi\ge k^2r^4,\qquad
\Psi/S^3=\Omega(k),\qquad
J=O(k^2)=o(S),
\]
the last bound following because those $k^2$
vertices already carry enough degree to reach
the threshold $S^2/256$.

For an arbitrary forward core of size $m_s$
and maximum degree $D_s$, all but
$B=r+k+2$ retained edges lie in $U\times W$,
so $m_s\le kD_s+B$. Symmetrically,
$m_t\le kD_t+B$ for any backward core.
If either core has at most $2B$ edges, its
objective is at most $2B$; otherwise,
$D_i>m_i/(2k)$ on both sides, whence
$m_sm_t/(D_sD_t)<4k^2$. Therefore,
\[
\mathcal L(E_s,E_t)\le\max\{2B,4k^2\}
=O(k^2)=o(S).
\]
Finally, for $D<k+1$ the degree tail $h_D$
contains all $2r$ vertices in $W\cup Z$;
for $k+1\le D<r+1$ it contains at least
the $2k$ vertices in $U\cup V$; and for
$D\ge r+1$ we have $D^2/n=\Omega(r)$.
Hence $\min_{D\ge1}(h_D+D^2/n)=\Omega(k)$.
The weight-critical certificate, in contrast,
gives $T_A=\Omega(|\mathcal C_{s,t}|)=\Omega(S)$.
\end{proof}

\paragraph{No converse from the critical-edge count.}
Take two disjoint $r$-cliques containing $s,a$
and $t,b$, respectively, joined by the sole
bridge $ab$, and give every edge weight $1$.
The shortest $st$-distance is $3$, and only
$sa,ab,bt$ are weight-critical: for an
internal edge not on that path, decreasing
its weight alone cannot beat length $3$.
Choose $a,b$ last among the distance-$1$
vertices in their respective searches.
Then BID explores $\Theta(r)$ clique vertices
of degree $\Theta(r)$ before finding the bridge,
so $n=2r$, $|\mathcal C_{s,t}|=3$, and
$S=\Theta(r^2)=\Theta(n^2)$. Thus $S$
has no upper bound in terms of the critical-edge
count alone.

\end{document}